\documentclass{article}
\usepackage{fullpage}

\usepackage[utf8]{inputenc}
\usepackage{amsthm,amssymb}
\usepackage{amsmath}
\usepackage{mathtools}
\usepackage{booktabs}
\usepackage{tcolorbox}
\usepackage{thmtools}
\usepackage{subcaption,thm-restate}
\usepackage[mathcal,mathscr]{eucal}
\usepackage{epsfig,graphicx,graphics,color}
\usepackage{caption}
\usepackage{enumerate,enumitem}
\usepackage{algorithm}
\usepackage{algpseudocode}
\usepackage[sort]{cite}
\usepackage{titling}
\usepackage{hyperref}
\usepackage{xspace}
\usepackage{tikz}
\usetikzlibrary{arrows,backgrounds,positioning,calc,fit,shapes.geometric,decorations.pathreplacing,decorations.markings, tikzmark,topaths}
\usepackage{tkz-euclide}

\tikzstyle{internal} = [draw, fill, shape=circle, text=black, inner sep=4pt]
\tikzstyle{external} = [shape=circle, draw]
\tikzstyle{square}   = [draw, fill, rectangle, inner sep=5pt]
\tikzset{lnode/.style = {
		circle, 
		draw=cyan!30!black, 
		thick,
		inner sep=1.5pt,
		minimum size=15pt } }
\tikzset{uedge/.style = {
		draw=cyan!20!black, 
		very thick} }
\hypersetup{
	colorlinks=true,       
	linkcolor=blue,        
	citecolor=red,         
	filecolor=magenta,     
	urlcolor=cyan,         
	linktocpage=true
}
\usepackage{tocloft}
\theoremstyle{plain}
\newtheorem{thm}{Theorem}[section]
\newtheorem{lem}[thm]{Lemma}

\newtheorem{claim}[thm]{Claim}

\newtheorem{prop}[thm]{Proposition}

\theoremstyle{definition}
\newtheorem{defn}[thm]{Definition}

\newtheorem{ex}[thm]{Example}

\newtheorem{rem}[thm]{Remark}

\def\final{0}  
\def\iflong{\iffalse}
\ifnum\final=0  
\newcommand{\kanote}[1]{{\color{blue}[{\tiny \textbf{Karthik:} \bf #1}]\marginpar{\color{blue}*}}}
\newcommand{\kristof}[1]{{\color{red}[{\tiny \textbf{Kristóf:} \bf #1}]\marginpar{\color{red}*}}}
\newcommand{\tnote}[1]{{\color{blue}[{\tiny \textbf{Tamás:} \bf #1}]\marginpar{\color{blue}*}}}
\newcommand{\dnote}[1]{{\color{orange}[{\tiny \textbf{Dani:} \bf #1}]\marginpar{\color{orange}*}}}
\else 
\newcommand{\kristof}[1]{}
\newcommand{\tnote}[1]{}
\newcommand{\dnote}[1]{}
\fi

\DeclareMathOperator\head{head}

\DeclareMathOperator*{\argmax}{arg\,max}
\DeclareMathOperator*{\argmin}{arg\,min}

\newcommand{\bR}{\mathbb{R}}
\newcommand{\bZ}{\mathbb{Z}}

\newcommand{\cE}{\mathcal{E}}
\newcommand{\cP}{\mathcal{P}}

\newcommand{\cS}{\mathcal{S}}
\newcommand{\cQ}{\mathcal{Q}}

\newcommand{\Z}{\mathbb{Z}}

\newcommand{\newfontobj}[2]{
	\newcommand{#1}[1]{
		\expandafter\def\csname##1\endcsname{{#2 ##1}}}}

\newfontobj{\class}{\rm} 

\class{PSPACE}
\class{L}
\class{SL}
\class{BPL}
\class{RL}
\class{NC}
\class{ZPL}
\class{NPSPACE}
\class{ASPACE}
\class{NL}
\class{coNL}
\class{EXP}
\class{NEXP}
\class{coNEXP}
\class{NE}
\class{E}
\class{AM}
\class{MA}
\class{NP}
\class{DNP}
\class{UP}
\class{P}
\class{RP}
\class{BPP}
\class{ZPP}
\class{EXPSPACE}
\class{coNP}
\class{coRP}
\class{coAM}
\class{IP}
\class{PCP}
\class{AC}
\class{TC}
\class{NC}
\class{MIP}

\class{BP}

\newtcolorbox{probbox}{arc=6pt,
                      colback=white!100,
                      colframe=black!50,
                      before skip=6pt,
                      after skip=6pt,
                      boxsep=1pt,
                      left=6pt,
                      right=6pt,
                      top=4pt,
                      bottom=4pt}

\newcommand{\searchprob}[3]{\small
   \begin{center}%
    \begin{minipage}{\linewidth}%
      \begin{probbox}
      \textsc{#1}\\[0.4ex]
      \textbf{Input:} #2\\[0.4ex]
      \textbf{Goal:} #3
      \end{probbox}
    \end{minipage}%
  \end{center}
}

\newcommand{\mcp}{\mathcal{P}}
\newcommand{\mcq}{\mathcal{Q}}
\newcommand{\mcr}{\mathcal{R}}
\newcommand{\mcs}{\mathcal{S}}
\newcommand{\R}{\mathbb{R}}

\newcommand{\graphkcut}{\ensuremath{\text{Graph-}k\text{-Cut}}\xspace}
\newcommand{\stgraphkcut}{\ensuremath{\{s,t\}\text{-Sep-Graph-}k\text{-Cut}}\xspace}
\newcommand{\kpartition}{\ensuremath{\text{Submod-}k\text{-Part}}\xspace}
\newcommand{\stkpartition}{\ensuremath{\{s,t\}\text{-Sep-Submod-}k\text{-Part}}\xspace}

\usepackage{accents}
\newcommand{\dG}{\mathpalette\dGaux G}
\newcommand{\dGaux}[2]{%
  \vphantom{#2}%
  \overset{%
    \smash{%
      \raisebox{%
        \ifx#1\scriptstyle -0.4ex 
        \else\ifx#1\scriptscriptstyle -0.2ex 
        \else -0.6ex 
        \fi\fi
      }{$\scriptstyle\rightharpoonup$}%
    }%
  }{#2}%
}
\newcommand{\dGp}{\mathpalette\dGaux {G'}}

\newlength{\bibitemsep}
\newlength{\bibparskip}
\let\oldthebibliography\thebibliography
\renewcommand\thebibliography[1]{%
  \oldthebibliography{#1}%
  \setlength{\parskip}{\bibitemsep}%
  \setlength{\itemsep}{\bibparskip}%
}

\makeatletter
\renewcommand{\paragraph}{%
  \@startsection{paragraph}{4}%
  {\z@}{1.6ex \@plus 1ex \@minus .2ex}{-0.5em}%
  {\normalfont\normalsize\bfseries}%
}
\makeatother

\let\Right\bigr
\let\Left\bigl
\makeatletter
\def\bigr#1{\Right#1\@ifnextchar){\!\bigr}{}}
\def\bigl#1{\Left#1\@ifnextchar({\!\bigl}{}}
\makeatother

\allowdisplaybreaks

\title{$\{s,t\}$-Separating Principal Partition Sequence\\ of Submodular Functions}
\date{}

\thanksmarkseries{alph}
\author{
Kristóf Bérczi\thanks{MTA-ELTE Matroid Optimization Research Group and HUN-REN–ELTE Egerváry Research Group, Department of Operations Research, ELTE Eötvös Loránd University, and HUN-REN Alfréd Rényi Institute of Mathematics, Budapest, Hungary. Email: \texttt{kristof.berczi@ttk.elte.hu}.}
\and
Karthekeyan Chandrasekaran\thanks{Grainger College of Engineering, University of Illinois, Urbana-Champaign. Email: \texttt{karthe@illinois.edu}.}
\and
Tamás Király\thanks{HUN-REN-ELTE Egerv\'ary Research Group, Department of Operations Research, E\"otv\"os Loránd University, Budapest, Hungary. Email: \texttt{tamas.kiraly@ttk.elte.hu}.}
\and
Daniel P. Szabo\thanks{Department of Operations Research, E\"otv\"os Loránd University, Budapest, Hungary. Email: \texttt{dszabo2@wisc.edu}.}
}

\begin{document}
	
	\maketitle
\vspace{-0.8cm}    


\begin{abstract}
Narayanan showed the existence of the the principal partition sequence of a submodular function, a structure with numerous applications in areas such as clustering, fast algorithms, and approximation algorithms.
In this work, motivated by two applications, we develop a theory of $\{s,t\}$-separating principal partition sequence of a submodular function. We define this sequence, show its existence, and design a polynomial-time algorithm to construct it. We show two applications: (1) approximation algorithm for the $\{s,t\}$-separating submodular $k$-partitioning problem for monotone and posimodular functions and (2) polynomial-time algorithm for the hypergraph orientation problem of finding an orientation that simultaneously has strong connectivity at least $k$ and $(s,t)$-connectivity at least $\ell$.

\medskip

\noindent \textbf{Keywords:} Submodular functions, Principal Partition Sequence, Submodular Partitioning, Hypergraph Orientation

\end{abstract}
\thispagestyle{empty}
\tableofcontents
\newpage
\pagenumbering{arabic}
\setcounter{page}{1}
\section{Introduction}
\label{sec:intro}


A set function $f\colon 2^V\to \bR$ is {\it submodular} if it satisfies the inequality $f(A)+f(B)\geq f(A\cap B)+f(A\cup B)$ for all $A,B\subseteq V$. Submodular functions arise throughout combinatorial optimization and economics; common examples include graph and hypergraph cut function, matroid rank function, and coverage function. 
The {\it principal partition sequence} 
is a central tool in submodular optimization, originating from the study of principal partitions by Kishi and Kajitani~\cite{kishi1969maximally} in the context of graph tri-partitions defined via maximally distant spanning trees. Initially viewed as decompositions of discrete systems into partially ordered components, principal partitions were later recognized to have a natural foundation in submodularity. Building on this foundation, Narayanan~\cite{pps} developed the theory of the principal partition sequence of submodular functions, while Fujishige~\cite{fujishige2009theory} provided a comprehensive survey of the theory of principal partitions. Over time, the theory of principal partition and principal partition sequence has been extended well beyond graphs~\cite{ozawa1974common} to matrices~\cite{iri1968min,iri1969maximum}, matroids~\cite{bruno1971principal,narayanan1974theory,tomizawa1976strongly}, and general submodular systems~\cite{fujishige1980lexicographically,fujishige1980principal,iri1979review,iri1984structural,nakamura1988structural, narayanan-book}, offering a lattice-theoretic decomposition framework for submodular functions. 

Principal partition sequence is a well-structured sequence of partitions of the ground set minimizing $\sum_{A\in \cP}f(A)-\lambda\cdot|\cP|$ as $\lambda$ varies from $-\infty$ to $\infty$ (see Section \ref{sec:stspps} for a formal definition). The sequence exhibits a refinement structure for submodular functions. 
A principal partition sequence of a submodular function given by its evaluation oracle can be found in polynomial time \cite{pps, narayanan-book, kolmogorovFast}. This is in contrast to minimizing $\sum_{A\in \cP}f(A)$ over partitions $\cP$ of the ground set with a given number of parts, which is $\NP$-hard. The computational tractability has made the principal partition sequence a powerful structural tool in submodularity with many applications. Examples include Cunningham's network strength measure that is used to quantify network vulnerability \cite{cunningham}, the realization of finite state machines~\cite{FSM-app}, recursive ideal tree packing \cite{CQX20}, approximation algorithms for graph and submodular partitioning problems \cite{Bar00, RS07, NRP96, ppskarthikwang}, graph clustering \cite{PN03, nagano2010minimum}, and dense subgraph decomposition \cite{fujishige1980lexicographically, hqc22, CCK25}.

In this paper, we extend the principal partition sequence to incorporate separation of two designated terminals -- say $s,t\in V$. In particular, we show that the sequence of $\{s,t\}$-separating partitions minimizing $\sum_{A\in \cP}f(A)-\lambda\cdot|\cP|$ as $\lambda$ varies from $-\infty$ to $\infty$ is well-structured and can be found in polynomial-time. The formal definition of $\{s,t\}$-separating principal partition sequence is fairly intricate, so we postpone its definition for now. Here, we discuss two concrete applications of this framework to illustrate its power. 

\subsection{Applications}

The first application concerns approximation algorithms and the second concerns polynomial-time algorithms. 

\subsubsection{\texorpdfstring{$\{s,t\}$}{{s,t}}-Separating Submodular \texorpdfstring{$k$}{k}-Partition}

Given a submodular function $f:2^V\rightarrow \R$ via its evaluation oracle, the {\it submodular $k$-partition} problem (abbreviated \kpartition) asks for a partition $\cP=\{V_1,\ldots,V_k\}$ of the ground set $V$ into $k$ non-empty parts that minimizes $f(\cP)\coloneqq \sum_{i=1}^k f(V_i)$. 
Many problems in combinatorial optimization can be formulated as special cases, including 
partitioning problems over graphs, hypergraphs, matrices, and matroids. 
\kpartition is NP-hard \cite{GH94}, does not admit a $(2-\epsilon)$-approximation assuming polynomial number of function evaluation queries \cite{San21}, does not admit a $n^{1/(\log\log n)^c}$-approximation for every constant $c$ assuming the Exponential Time Hypothesis \cite{k-wayhypergraphcuthardness}, and the best approximation factor that is known is $O(k)$ \cite{greedysplit, submod_kapx}. 
Nevertheless, constant factor approximations are known for broad subfamilies of submodular functions. 

We recall that a function $f$ is \emph{symmetric} if $f(S)=f(V\setminus S)$ for all $S\subseteq V$, \emph{monotone} if $f(S)\le f(T)$ for every $S\subseteq T\subseteq V$, and \emph{posimodular} if $f(S) + f(T) \ge f(S\setminus T) + f(T\setminus S)$ for every $S, T\subseteq V$. We observe that symmetric/monotone submodular functions are also posimodular. Prominent examples of symmetric submodular functions include graph and hypergraph cut functions, monotone submodular functions include matroid rank functions and coverage functions, and posimodular functions include positive combinations of symmetric submodular and monotone submodular functions. 
A well-studied special case of \kpartition for symmetric submodular functions is \graphkcut: the input is a graph and the goal is to find a minimum number of edges to delete so that the resulting graph has at least $k$ components. 
If $k$ is a fixed constant, then \graphkcut is polynomial-time solvable \cite{GH94}. For $k$ part of input, \graphkcut is NP-complete \cite{GH94} and does not have a polynomial-time $(2-\epsilon)$-approximation for every constant $\epsilon>0$ under the Small Set Expansion Hypothesis \cite{Ma18}. 
There are several approaches that achieve a $2$-approximation for \graphkcut \cite{RS07, Bar00, SV95, NK07}. 
One of these approaches is the principal partition sequence, which has been generalized to achieve a $2$-approximation for posimodular submodular $k$-partition and a $4/3$-approximation for monotone submodular $k$-partition \cite{ppskarthikwang}. 


Our first application of $\{s,t\}$-separating principal partition sequence is to the submodular $k$-partition problem with the additional constraint that the minimizing partition separates two specified terminals. We term this as \stkpartition and formally define it below. For a ground set $V$ and a pair of terminals $s,t\in V$, a partition $\mcp$ of $V$ is \textit{$\{s,t\}$-separating} if $|P\cap \{s,t\}|\le 1$ for every $P\in \mcp$. 
 {\searchprob{\stkpartition}{A submodular function $f\colon 2^V\to \bR_{\ge 0}$ given by a value oracle, distinct elements $s, t \in V$, and $k\in \Z_{\geq 0}$.}{
\[ \min\left\{ \sum_{i=1}^k f(V_i) \colon  \{ V_i\}_{i=1}^k \text{ is a }\{s,t\}\text{-separating partition of } V \text{ into $k$ non-empty parts}\right\}.\]}}





A concrete special case of \stkpartition is \stgraphkcut: the input is a graph with two specified terminal vertices $s,t$ and the goal is to find a minimum number of edges to delete so that the resulting graph has at least $k$ components with $s$ and $t$ being in different components. 
If $k$ is a fixed constant, \stgraphkcut is solvable in polynomial time \cite{berczi2019beating}. Thus, for fixed constant $k$, the complexity status of \stgraphkcut is identical to that of \graphkcut. This status for \stgraphkcut for constant $k$ inspired us to investigate \stgraphkcut, and more generally, \stkpartition, when $k$ is part of input. 
Is it possible to achieve the same approximation results for \stkpartition as for \kpartition? 

We observe that \stkpartition closely resembles \kpartition, but known approaches for \kpartition do not apply directly. It is easy to see that \stkpartition is a special case of matroid constrained submodular $k$-partition and consequently, it admits a $2$-approximation for symmetric submodular functions via the Gomory-Hu tree approach \cite{berczi2025approximating}. In this work, we exploit $\{s,t\}$-separating principal partition sequence to design approximation algorithms for \stkpartition with an approximation guarantee that matches that of \kpartition for monotone and posimodular submodular functions. 

\begin{thm}
    \stkpartition admits a $2$-approximation for posimodular submodular functions and a $4/3$-approximation for monotone submodular functions. 
\end{thm}
We note that the guarantee for posimodular submodular functions implies the same guarantee for symmetric submodular functions. Consequently, we have a $2$-approximation for \stgraphkcut. 

\newcommand{\stellConnOrient}{\ensuremath{(s,t)\text{-}\ell\text{-Conn-Orient}}\xspace}
\newcommand{\kConnOrient}{\ensuremath{k\text{-Conn-Orient}}\xspace}
\newcommand{\kstellConnOrient}{\ensuremath{(k,(s,t),\ell)\text{-Conn-Orient}}\xspace}
\newcommand{\graphstellConnOrient}{Graph-\ensuremath{(s,t)\text{-}\ell\text{-Conn-Orient}}\xspace}
\newcommand{\graphkConnOrient}{Graph-\ensuremath{k\text{-Conn-Orient}}\xspace}
\newcommand{\graphkstellConnOrient}{Graph-\ensuremath{(k,(s,t),\ell)\text{-Conn-Orient}}\xspace}

\subsubsection{Hypergraph Orientation}
In graph orientation problems, the goal is to orient a given undirected graph to obtain a directed graph with certain properties. Two fundamental properties of interest are $(s,t)$-connectivity and strong-connectivity. In \graphstellConnOrient, the input is an undirected graph with two specified terminals $s$ and $t$ and an integer $\ell\in \Z_{\geq 0}$. The goal is to verify if there exists an orientation of $G$ that has $\ell$ arc-disjoint paths from $s$ to $t$ and if so, find one. By Menger's theorem, an undirected graph $G$ has an orientation with $\ell$ arc-disjoint paths from $s$ to $t$ if and only if there exist $\ell$ edge-disjoint paths between $s$ and $t$ in the undirected graph $G$. Thus, \graphstellConnOrient can be solved using an application of max $(s,t)$-flow. 
In \graphkConnOrient, the input is an undirected graph and an integer $k\in \Z_{\geq 0}$. 
The goal is to verify if there exists a $k$-arc-connected orientation of the graph. 
We recall that a digraph is $k$-arc-connected if for every pair of distinct vertices $u, v\in V$, there exist $k$ arc-disjoint paths from $u$ to $v$. A classic result of Nash-Williams \cite{nash1960orientations} shows that an undirected graph $G$ has a $k$-arc-connected orientation if and only if $G$ is $2k$-edge-connected. We recall that an undirected graph $G=(V, E)$  is $2k$-edge-connected if for every pair of distinct vertices $u, v\in V$, there exist $2k$ edge-disjoint paths between $u$ and $v$. Nash-Williams' result is also constructive that leads to a polynomial time algorithm to solve \graphkConnOrient. 

Frank, Kir\'{a}ly, and Kir\'{a}ly \cite{tamas-frank-zoli} showed an intriguing result concerning a generalization of these two problems that we term as \graphkstellConnOrient. The input here is an undirected graph with two specified terminals $s$ and $t$ and integers $k, \ell\in \Z_{\geq 0}$ (with the interesting case being $k<\ell$). The goal is to verify if there exists an orientation $\dG$ of $G$ such that $\dG$ is $k$-arc-connected and has $\ell$ arc-disjoint paths from $s$ to $t$ -- we will call such an orientation to be $(k, (s,t), \ell)$-arc-connected orientation. They showed a complete characterization for the existence of a $(k, (s,t),\ell)$-arc-connected orientation while also giving a polynomial time algorithm to find such an orientation via Mader's splitting-off result. Our focus is on this generalized orientation problem over hypergraphs instead of graphs. 

Frank, Kir\'{a}ly, and Kir\'{a}ly \cite{tamas-frank-zoli} considered all three orientation problems mentioned above in hypergraphs as opposed to graphs. We begin by discussing the two constituent problems in hypergraphs. 
A hypergraph $G=(V, E)$ is specified by a vertex set $V$ and hyperedge set $E$ where each $e\in E$ is a subset of $V$. 
An orientation $\dG$ of a hypergraph $G=(V, E)$ is a directed hypergraph specified by $\dG=(V, E, \head\colon E\rightarrow V)$, where $\head(e)\in e$ for every $e\in E$. 
In hypergraph orientation problems, the goal is to orient a given hypergraph to obtain a directed hypergraph with certain properties. 
For vertices $a, b$, an $(a,b)$-path in a directed hypergraph is an alternating sequence, without repetition, of
vertices and hyperarcs $v_1, e_1, v_2, e_2, \ldots, v_k, e_k, v_{k+1}$ where $v_i\in e_i\setminus \head(e_i)$ and $v_{i+1}=\head(e_i)$. 
In \stellConnOrient, the input is a hypergraph with two specified terminals $s$ and $t$ and an integer $\ell\in \Z_{\geq 0}$. The goal is to verify if there exists an orientation of $G$ that has $\ell$ hyperarc-disjoint paths from $s$ to $t$ and if so, then find one. Menger's theorem extends to hypergraphs and thus, resolves $\stellConnOrient$ in hypergraphs. 
In \kConnOrient, the input is a hypergraph and an integer $k\in \Z_{\geq 0}$. 
The goal is to verify if there exists a $k$-hyperarc-connected orientation of the graph. A directed hypergraph is $k$-hyperarc-connected if for every pair of distinct vertices $u, v\in V$, there exist $k$ hyperarc-disjoint paths from $u$ to $v$. 
Frank, Kir\'{a}ly, and Kir\'{a}ly generalized Nash-Williams' result from graphs to hypergraphs by giving a complete characterization for the existence of a $k$-hyperarc-connected orientation of a hypergraph and moreover, their characterization is constructive, i.e., it leads to a polynomial-time algorithm to verify if a given hypergraph has a $k$-hyperarc-connected orientation and if so, then find one. 

We now address the main application of interest to our work, namely \kstellConnOrient that is defined as follows (with the interesting case being $k<\ell$). 
 {\searchprob{\kstellConnOrient}{A hypergraph $G=(V, E)$, vertices $s, t \in V$, and $k, \ell\in \Z_{\geq 0}$.}{
Verify if there exists an orientation $\dG$ of $G$ such that $\dG$ is $k$-hyperarc-connected and has $\ell$ hyperarc-disjoint paths from $s$ to $t$.}}
We will call a feasible orientation for \kstellConnOrient to be a $(k, (s,t), \ell)$-hyperarc-connected orientation. Frank, Kir\'{a}ly, and Kir\'{a}ly gave a characterization for the existence of a $(k, (s,t),\ell)$-connected orientation of a hypergraph. However, their proof of the characterization 
relies on certain steps which did not seem constructive. Consequently, it was unclear whether there exists a polynomial-time algorithm to verify 
whether a given hypergraph has a $(k, (s,t),\ell)$-hyperarc-connected orientation and if so, then find it. In this work, we exploit algorithmic aspects of $\{s,t\}$-separating principal partition sequence of a submodular function to design a polynomial time algorithm for this problem. 

\begin{thm}
    There exists a polynomial time algorithm to verify whether a given hypergraph admits a $(k, (s,t),\ell)$-hyperarc-connected orientation and if so, then find one. 
\end{thm}
We also consider optimization variants of the above problem where $\ell$ (or $k$) is given as part of input and the goal is to maximize $k$ (or $\ell$ respectively) so that there exists a $(k, (s,t),\ell)$-connected orientation, present min-max relations, and show that the corresponding minimization problem is solvable in polynomial time via our results on $\{s,t\}$-separating principal partition sequence. 




\subsection{\{s,t\}-Separating Principal Partition Sequence}
\label{sec:stspps}

In this section, we present our definition and state our results concerning $\{s,t\}$-separating principal partition sequence. We begin by discussing principal partition sequence and known results to provide background and context. 
For a set function $f\colon 2^V\to\bR$ and a collection $\mcp$ of subsets of $V$, we write $f(\mcp)\coloneqq \sum_{P\in \mcp}f(P)$. A \textit{partition} of a set $S\subseteq V$ is a collection of pairwise disjoint non-empty subsets of $S$ whose union is $S$. 

\paragraph{Principal Partition Sequence.}
Let $f\colon 2^V\to\bR$ be a set function and $\mcp$ be a partition of the ground set $V$. 
We define functions $g_{f,\mcp}:\bR\to\bR$ and $g_f\colon \bR\to\bR$ as follows:
\begin{align*}
    g_{f,\mcp}(\lambda)&\coloneqq f(\mcp)-\lambda\cdot |\mcp|, \\
    g_f(\lambda)&\coloneqq \min\{g_{f,\mcp}(\lambda):\mcp\text{ is a partition of }V\}.
\end{align*}
We drop the subscript $f$ and simply write $g_{\mcp}$ and $g$, respectively, if the function $f$ is clear from context. A \emph{breakpoint} of a piecewise linear function is a point at which the slope of $f$ changes, or equivalently, a point where $f$ is continuous but not differentiable. 
By definition, the function $g_f$ is piecewise linear and has at most $|V|$ slopes and hence, at most $|V|-1$ breakpoints.

Narayanan~\cite{narayanan1974theory}  showed that if $f\colon 2^V\to \bR$ is submodular, then the partitions attaining the values of the function $g_f$ at its breakpoints exhibit a refinement structure, which is captured by the principal partition sequence. 

\begin{defn}[Refinement]\label{def:refine}
    Let $\cP$ and $\cQ$ be distinct partitions of $V$. 
    \begin{enumerate}\itemsep0em
        \item $\mcq$ is a \emph{refinement of $\mcp$} if for all $B\in \cQ$ there exists $A\in \cP$ such that $B\subseteq A$. \label{refinment:1}
        \item $\mcq$ is a refinement of $\mcp$ \emph{up to one set} if 
        \begin{enumerate}\itemsep0em
            \item $\mcq$ is a refinement of $\mcp$ and \label{refinment:2a}
            \item there exists $P\in \mcp$ such that for all $B\in \mcq$, either $B\subsetneq P$ or $B\in \mcp$. \label{refinment:2b}
        \end{enumerate} \label{refinment:2}
    \end{enumerate}
\end{defn}

See Figure~\ref{fig:refinement} for an example of a refinement. We observe that if $\mcp$ and $\mcq$ are distinct partitions and $\mcq$ is a refinement of $\mcp$, then $|\mcq|>|\mcp|$. 
Using this definition of refinements, a principal partition sequence is defined as follows.

\begin{defn}[Principal partition sequence] \label{def:pps}
Let $f\colon 2^V\to\bR$ be a set function. 
A sequence $\mcp_1, \dots, \mcp_{\ell}$ of partitions of $V$ is a \emph{principal partition sequence of $f$} if there exist $\lambda_1,  \ldots, \lambda_{\ell-1}\in\bR$ such that 
\begin{enumerate}[label=(P\arabic*)]\itemsep0em
    \item $\lambda_1\le \ldots \le \lambda_{\ell-1}$, \label{principal:1}
    \item 
    $g(\lambda)=\begin{cases}
        g_{\mcp_1}(\lambda) & \text{for $\lambda\in (-\infty, \lambda_1]$},\\
        g_{\mcp_{j+1}}(\lambda) & \text{for $\lambda\in [\lambda_j, \lambda_{j+1}]$ and $j\in [\ell-2]$},\\
        g_{\mcp_{\ell}}(\lambda) & \text{for $\lambda\in [\lambda_{\ell-1}, \infty)$},
    \end{cases}$\label{principal:2}
    \item $\mcp_1=\{V\}$ and $\mcp_\ell=\{\{v\}\colon v\in V\}$, and \label{principal:3}
    \item for every $j\in [\ell-1]$,  $\mcp_{j+1}$ is a refinement of $\mcp_j$ up to one set. \label{principal:4}
\end{enumerate}
The values $\lambda_1,  \ldots, \lambda_{\ell-1}$ are called the \emph{critical values associated with the principal partition sequence}.
\end{defn}

We observe that properties \ref{principal:1}, \ref{principal:2}, and \ref{principal:3} in the definition of a principal partition sequence are easy to achieve for every set function. The essence of principal partition sequences lies in the fact that, for submodular functions, property~\ref{principal:4} can also be achieved. This was shown by Narayanan~\cite{narayanan1974theory}. He also designed a polynomial-time algorithm to find such a sequence. 

\begin{thm}[Narayanan]\label{thm:PPS-exists-algo}
Let $f:2^V\rightarrow \R$ be a submodular function. Then, there exists a principal partition sequence of $f$. Moreover, given oracle access to a submodular function $f\colon 2^V\to \bR$, there exists a polynomial-time algorithm to compute a principal partition sequence of $f$ along with its associated critical values.
\end{thm}

Principal partition sequence resembles the notion of \emph{principal partition} of a submodular function. We will not define principal partition of a submodular function in this work since this is not directly relevant to our investigation. It is known that the principal partition sequence of the cut function of a graph is related to the principal partition of the rank function of the graphic matroid associated with the graph. We refer the reader to \cite{narayanan-book} and \cite{fujishige2009theory} for more details on similarities and differences between the two notions.

\paragraph{$\{s,t\}$-Refinement and $\{s,t\}$-Separating Principal Partition Sequence.}
Let $f\colon  2^V\rightarrow \bR$ be a submodular set function on ground set $V$ 
and $s, t\in V$ be distinct elements. 
A partition $\mcp$ of $V$ is an $\{s,t\}$-separating partition if $|\{s,t\}\cap P|\le 1$ for every $P\in \mcp$. 
We recall that for every partition $\mcp$ of $V$ and every $\lambda\in\bR$, we defined $g_{f,\mcp}(\lambda)=f(\mcp)-\lambda\cdot |\mcp|$. 
As a natural $\{s,t\}$-separating counterpart, we define 
the function $g_f^{s,t}\colon \bR\to\bR$ by setting
\begin{align*}
    g_f^{s,t}(\lambda)&\coloneqq \min\{g_{f,\mcp}(\lambda)\colon \mcp\text{ is an $\{s,t\}$-separating partition of }V\}.
\end{align*}
We drop the subscript $f$ and simply write $g^{s,t}$ if the function $f$ is clear from context. The function $g^{s,t}$ is piecewise linear with at most $|V|-2$ breakpoints (see Lemma \ref{lem:piecewise-linear} in appendix).

We observe that for every set function $f\colon 2^V\to\bR$, one can find a sequence $\cP_1,\dots,\cP_\ell$ of partitions and $\lambda_1,\dots,\lambda_\ell\in\bR$ satisfying properties \ref{principal:1}-\ref{principal:3} in Definition~\ref{def:pps}. Submodularity becomes crucial, however, in order to obtain the refinement property \ref{principal:4}. A key difference between principal partition and the $\{s,t\}$-separating partition sequences arises precisely in this property: minimizers for the breakpoints of $g^{s,t}$ do not necessarily satisfy the refinement property; an example is given in Figure~\ref{fig:example2}. 

\begin{figure}[ht]
    \centering
    \begin{tikzpicture}[scale=1.2, 
    vertex/.style={line width=1.5pt, circle, draw, fill=white, inner sep=3pt},
    edgelabel/.style={fill=none, inner sep=1pt, font=\small\itshape},
    every edge/.append style = {draw, thick}]        
        \node[vertex] (a) at (-1,1) {$s$};
        \node[vertex] (b) at (-1,-1) {$a$};
        \node[vertex] (c) at (0,0) {$b$};
        \node[vertex] (e) at (1,0) {$c$};
        \node[vertex] (f) at (2,1) {$d$};
        \node[vertex] (g) at (3,0) {$t$};
        \node[vertex] (h) at (2,-1) {$e$};
        
        \draw (a) -- node[edgelabel, left=1pt] {$\alpha_1$} (b);
        \draw (b) -- node[edgelabel, below right] {$\alpha_1$} (c);
        \draw (c) -- node[edgelabel, above right] {$\alpha_1$} (a);
            \draw (c) -- node[edgelabel, above=1pt] {$1$} (e);
        \draw (e) -- node[edgelabel, above left] {$\alpha_2$} (f);
        \draw (f) -- node[edgelabel, above right] {$\alpha_2$} (g);
        \draw (g) -- node[edgelabel, below right] {$\alpha_2$} (h);
        \draw (h) -- node[edgelabel, below left] {$\alpha_2$} (e);
        \draw (e) -- (g);
        \draw (f) -- (h);
        \node[edgelabel] at (1.57, 0.1) {$\alpha_2$};
        \node[edgelabel] at (2.15, -0.4) {$\alpha_2$};
    \end{tikzpicture}
    \caption{An edge-weighted graph whose cut function $f$ is such that minimizers for the breakpoints of $g_f^{s,t}$ do not necessarily satisfy the refinement property. 
    Here $0<\varepsilon<1/24$ is a small constant, and $\alpha_1 = 1/2+\varepsilon,$ $\alpha_2 = 1/3 + 2\varepsilon$. The sequence of minimizers of $g_f^{s,t}$ 
    is unique and is given by $\cP_1 = \{\{s,a,b\},\{t,c,d,e\}\}, \cP_2 = \{\{s\},\{a\},\{b,t,c,d,e\}\},$ $\cP_3:=\{\{s,a,b,c\}, \{d\}, \{e\}, \{t\}\},$ $\cP_4 = \{\{s\}, \{a\}, \{b, c\}, \{d\}, \{e\}, \{t\}\},$ and $\cP_5 = \{\{s\}, \{a\}, \{b\}, \{c\}, \{d\}, \{e\}, \{t\}\}$. The principal partition sequence, meanwhile, is $\{s,a,b,c,d,e,t\}, \cP_3, \cP_4, \cP_5$. }
    \label{fig:example2}
\end{figure}

Nevertheless, this sequence still exhibits a more intricate form of refinement. We now formalize this intricate form of refinement. Two sets $X\subseteq V$ and $Y\subseteq V$ are called \emph{intersecting} if all of $X\cap Y$, $X\setminus Y$, and $Y\setminus X$ are nonempty. A pair of intersecting sets $X\subseteq V$ and $Y\subseteq V$ is called {\it $\{s,t\}$-uncrossable} if they do not separate $s$ and $t$, meaning that either $|\{s,t\}\cap (X\setminus Y)|\neq 1$ or $|\{s,t\}\cap (Y\setminus X)|\neq 1$.

\begin{defn}[$\{s,t\}$-refinements]\label{def:strefine}
    Let $\cP$ and $\cQ$ be distinct $\{s,t\}$-separating partitions of $V$. 
    \begin{enumerate}
        \item $\mcq$ is a \emph{$\{s,t\}$-refinement of $\mcp$ along $(X,Y)$} if $X\in \mcp\setminus\mcq$, $Y\in \mcq\setminus \mcp$, and the following hold:
        \begin{enumerate}[label=(\alph*)]\itemsep0em
            \item $X$ and $Y$ are intersecting but they are \emph{not} $\{s,t\}$-uncrossable, \label{strefinement:1a}
            \item for all $B\in \cQ\setminus Y$ there exists $A\in \cP$ such that $B\subseteq A$,\label{strefinement:1b} 
            \item for all $A\in \cP \setminus X$, either $A\cap Y = \emptyset$ or $A\subseteq Y$, and \label{strefinement:1c}
            \item $|\{A\in \cP \colon A\subseteq Y\}| \leq |\{B\in \cQ \colon B\subseteq X\}|$. \label{strefinement:1d}
        \end{enumerate}\label{strefinement:1}
        \item $\mcq$ is a \emph{$\{s,t\}$-refinement of $\mcp$} if there exists $X\in \mcp\setminus \mcq$ and $Y\in \mcq\setminus \mcp$ such that $\mcq$ is an $\{s,t\}$-refinement of $\mcp$ along $(X,Y)$. \label{strefinement:2}
        \item $\mcq$ is a \emph{$\{s,t\}$-refinement of $\mcp$ up to two sets} if there exists $X\in \mcp\setminus \mcq$ and $Y\in \mcq\setminus \mcp$ such that 
        \begin{enumerate}\itemsep0em
            \item $\mcq$ is an $\{s,t\}$-refinement of $\mcp$ along $(X,Y)$ and 
            \item $\{P\in \mcp\colon P\subseteq V\setminus (X\cup Y)\}\subseteq \mcq$.
        \end{enumerate}\label{strefinement:3}
    \end{enumerate}
\end{defn}
See Figures~\ref{fig:refinement} and~\ref{fig:st-refinement} for examples of refinements and $\{s,t\}$-refinements. We observe that an $\{s,t\}$-refinement allows exactly one pair of intersecting sets in $\mcp \cup \mcq$, and these sets cannot be $\{s,t\}$-uncrossable. Moreover, if $\mcq$ is a refinement of $\mcp$, then $|\mcq| > |\mcp|$, whereas if $\mcq$ is an $\{s,t\}$-refinement of $\mcp$, then $|\mcq| \ge |\mcp|$.

\begin{figure}[t!]
\centering
\begin{subfigure}[t]{0.48\linewidth}
    \centering
    \includegraphics[width=\textwidth]{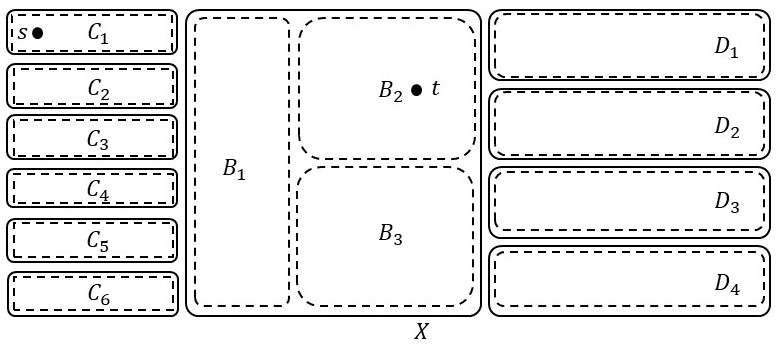}
\caption{An example where $\mcq$ is a refinement of $\mcp$, where $\mcp=\{X, C_1, C_2, C_3, C_4, C_5, C_6, D_1, D_2, D_3, D_4\}$ (solid lines) and $\mcq=\{B_1, B_2, B_3,  C_1, C_2, C_3, C_4, C_5, C_6, \allowbreak D_1, D_2, D_3, D_4\}$ (dashed lines).}
\label{fig:refinement}
\end{subfigure}
\hfill
\begin{subfigure}[t]{0.48\linewidth}
    \centering
    \includegraphics[width=\textwidth]{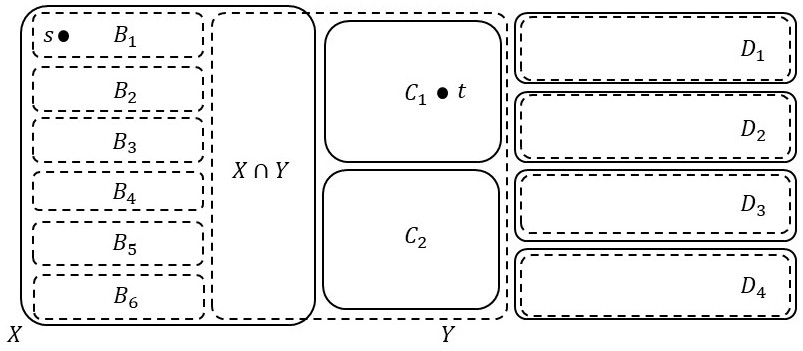}
    \caption{An example where $\mcq$ is an $\{s,t\}$-refinement of $\mcp$ up to $X$ and $Y$, where $\mcp=\{X, C_1, \allowbreak C_2, \allowbreak D_1, \allowbreak D_2,\allowbreak D_3, D_4\}$ (solid lines) and $\mcq=\{Y, B_1, B_2, \allowbreak B_3, \allowbreak B_4, \allowbreak B_5,  \allowbreak B_6, \allowbreak D_1, D_2, D_3, D_4\}$ (dashed lines).}
    \label{fig:st-refinement}    
\end{subfigure}
\caption{Examples of refinement and $\{s,t\}$-refinement.}
\label{fig:ref}
\end{figure}

Based on the notion of refinements and $\{s,t\}$-refinements, we define an $\{s,t\}$-separating principal partition sequence as follows. 

\begin{defn}[$\{s,t\}$-separating principal partition sequence]\label{def:st-pps}
    Let $f\colon 2^V\rightarrow \R$ be a set function. A sequence $\mcp_1, \ldots, \mcp_{\ell}$ of $\{s,t\}$-separating partitions of $V$ is a \emph{$\{s,t\}$-separating principal partition sequence} of $f$ if there exist $\lambda_1, \ldots, \lambda_{\ell-1}\in \R$  such that 
    \begin{enumerate}[label=(SP\arabic*)]
        \item $\lambda_1\le  \ldots \le \lambda_{\ell-1}$, \label{stpps:1}
        \item $g^{s,t}(\lambda)=\begin{cases}
        g_{\mcp_1}(\lambda) & \text{for $\lambda\in (-\infty, \lambda_1]$},\\
        g_{\mcp_{j+1}}(\lambda) & \text{for $\lambda\in [\lambda_j, \lambda_{j+1}]$ and $j\in [\ell-2]$},\\
        g_{\mcp_{\ell}}(\lambda) & \text{for $\lambda\in [\lambda_{\ell-1}, \infty)$}.
    \end{cases}$\label{stpps:2}
    \item $\mcp_1$ is $\{s,t\}$-separating partition with two parts such that $f(\mcp_1)\le f(A) + f(V\setminus A)$ for every $s\in A\subseteq V-t$,  and $\mcp_{\ell}=\{\{v\}\colon v\in V\}$, \label{stpps:3} 
    \item for every $j\in [\ell-1]$, either $\mcp_{j+1}$ is a refinement of $\mcp_j$ up to one set or $\mcp_{j+1}$ is an $\{s,t\}$-refinement of $\mcp_j$ up to two sets, and moreover, $|\mcp_{j}|<|\mcp_{j+1}|$. \label{stpps:4} 
    \end{enumerate}
    The values $\lambda_1,  \ldots, \lambda_{\ell-1}$ are called the \emph{critical values associated with the $\{s,t\}$-separating principal partition sequence}.
\end{defn}

We encourage the reader to compare and contrast Definitions~\ref{def:pps} and~\ref{def:st-pps}. 
A principal partition sequence can be viewed as a prefix of non-$\{s,t\}$-separating partition subsequence followed by a suffix of $\{s,t\}$-separating partition subsequence. However, the suffix by itself is not necessarily an $\{s,t\}$-separating principal partition sequence. 
We now state our main result regarding $\{s,t\}$-separating principal partition sequence of a submodular function.

\begin{thm}\label{thm:st-pps}
    Let $f\colon 2^V \to \bR$ be a submodular function and $s,t \in V$. Then, there exists an $\{s,t\}$-separating principal partition sequence of $f$. Moreover, given oracle access to $f$, such a sequence and its critical values can be computed in polynomial time.
\end{thm}

\section{\texorpdfstring{$\{s,t\}$}{(s,t)}-Separating Principal Partition Sequence}
\label{sec:stsep}

In this section, we prove Theorem \ref{thm:st-pps}, i.e., we prove the existence of an $\{s,t\}$-separating principal partition sequence of a submodular function and give a polynomial-time algorithm to construct it. 
We say that a submodular function $f: 2^V\rightarrow \bR$ is \emph{strictly submodular on intersecting pairs} if $f(A) + f(B) > f(A\cap B) + f(A\cup B)$ for every $A, B\subseteq V$ for which $A\cap B, A\setminus B, B\setminus A$ are non-empty.
We begin by showing Theorem \ref{thm:st-pps} for submodular functions that are strictly submodular on intersecting pairs, as stated below. 

\begin{thm}\label{lem:strict-st-pps}
    Let $f\colon 2^V \to \bR$ be a submodular function \emph{that is strictly submodular on intersecting pairs} and $s,t \in V$. Then, there exists an $\{s,t\}$-separating principal partition sequence of $f$. Moreover, given oracle access to $f$, such a sequence and its critical values can be computed in polynomial time.
\end{thm}

In Section~\ref{subsec:st-pps-refinement}, we establish a refinement property of $\{s,t\}$-separating partitions (Theorem~\ref{thm:refinement}) that plays a key role in proving the existence of the sequence. We build on this in Section~\ref{subsec:st-pps-existence} to show that an $\{s,t\}$-separating principal partition sequence exists  (Theorem~\ref{thm:stppsexists}). We present a polynomial-time algorithm to compute an $\{s,t\}$-separating principal partition sequence in Section~\ref{subsec:st-pps-algo} thereby completing the proof of Theorem \ref{lem:strict-st-pps}. 
We prove Theorem \ref{thm:st-pps} using Theorem \ref{lem:strict-st-pps} in Section \ref{sec:relaxing-strict-submodularity}.
\subsection{Refinement Property}
\label{subsec:st-pps-refinement}

We first show a refinement property that will be useful for proving the existence of an $\{s,t\}$-separating principal partition sequence. The main result is the following. 


\begin{thm}\label{thm:refinement}
    Let $f\colon 2^V\to \bR$ be a submodular function 
    that is strictly submodular on intersecting pairs, 
    and let $s,t\in V$. Suppose that $\mcp$ and $\mcq$ are distinct $\{s,t\}$-separating partitions such that at least one of the following holds: 
        \begin{enumerate}\itemsep0em
            \item There exists $\lambda\in \R$ such that 
            \begin{align*}
                \cP &\in \argmin \{|\pi|\colon \pi\text{ is an }\{s,t\}\text{-separating partition with }g_\pi(\lambda) = g^{s,t}(\lambda) \},\\
                \cQ &\in \argmax \{|\pi|\colon \pi\text{ is an }\{s,t\}\text{-separating partition with }g_\pi(\lambda) = g^{s,t}(\lambda) \}.
            \end{align*}
            \item There exists $\alpha < \lambda $ such that $g^{s,t}(\alpha)=g_{\mcp}(\alpha)$ and $g^{s,t}(\lambda) = g_{\mcp}(\lambda) = g_{\mcq}(\lambda)$. 
        \end{enumerate} 
    Then, $\cQ$ is either a refinement or an $\{s,t\}$-refinement of $\mcp$.
\end{thm}

    
    We begin with a few technical lemmas used in the proof of the theorem. The next lemma shows that every $\{s,t\}$-separating set system that covers each element exactly twice and contains more than two intersecting sets admits an $\{s,t\}$-uncrossable pair of intersecting sets. 
    
    \begin{lem}\label{lem:no_uncross}
        Let $\cS\subseteq 2^V$ be a set system such that
        \begin{enumerate}\itemsep0em
            \item\label{it:ex2} each $v\in V$ is contained in exactly two sets in $\cS$, and 
            \item $|\{s,t\}\cap A| \leq 1$ for all $A\in \cS$.
        \end{enumerate}
        Suppose that $\cS$ has more than two intersecting pairs of sets. Then, $\cS$ has a pair of intersecting sets that are $\{s,t\}$-uncrossable.
    \end{lem}
    \begin{proof}
        Suppose none of the intersecting pairs in $\mathcal{S}$ are $\{s,t\}$-uncrossable, and let $X_1,X_2\in \mathcal{S}$ be an intersecting pair. Then $|\{s,t\}\cap (X_1\setminus X_2)|=1=|\{s,t\}\cap (X_2\setminus X_1)|$. Without loss of generality let $s\in X_1\setminus X_2$ and $t\in X_2\setminus X_1$. Since $\cS$ has more than two intersecting pairs, consider another such pair $X_3, X_4 \in \cS$ with $\{X_1, X_2\}\neq \{X_3,X_4\}$. Similarly as before, we have $|\{s,t\}\cap (X_3 \setminus X_4)| = 1 = |\{s,t\}\cap (X_4 \setminus X_3)|$. We consider two cases.
        \medskip
        
        \noindent \textbf{Case 1.} $\{X_1,X_2\}\cap\{X_3,X_4\}=\emptyset$.
        
        Suppose that $s\in X_3\setminus X_4$ and $t\in X_4\setminus X_3$; the other case is symmetric, and an analogous argument applies. Then $s\in X_1 \cap X_3$ and $t\notin X_1\cup X_3$. Since $X_1$ and $X_3$ are not $\{s,t\}$-uncrossable, it follows that either $X_1\subseteq X_3$ or $X_3\subseteq X_1$, in which case elements of $X_1\cap X_2$ or $X_3\cap X_4$ belong to at least three sets -- $X_1, X_2, X_3$ or $X_1, X_3, X_4$, respectively -- yielding a contradiction.
        \medskip
        
        \noindent \textbf{Case 2:} $|\{X_1, X_2\}\cap \{X_3, X_4\}|=1$. 
        
        Without loss of generality, assume $X_1=X_4$. Since $s\in X_1=X_4$ and $|\{s,t\}\cap (X_3\setminus X_1)|=1$, it follows that $t\in (X_2\setminus X_1)\cap (X_3\setminus X_1)$. If $X_2$ and $X_3$ are intersecting, they would be $\{s,t\}$-uncrossable, since $(X_2\setminus X_3)\cap \{s,t\}=\emptyset$. Therefore, either $X_2\subseteq X_3$ or $X_3\subseteq X_2$, and consequently $X_1\cap X_2\cap X_3\neq \emptyset$, contradicting the assumption that each element is contained in exactly two sets in $\mathcal{S}$. 
    \end{proof}
    
    The following lemma shows that a pair of $g^{s,t}$-minimizing partitions cannot contain an $\{s,t\}$-uncrossable pair of intersecting sets. The proof crucially relies on the strict submodularity of the set function $f$ on intersecting pairs.
    
    \begin{lem}\label{lem:key_lem}
        Let $\lambda\in\R$, and let $\cP$ and $\cQ$ be $\{s,t\}$-separating partitions such that $g^{s,t}(\lambda)=g_{\cP}(\lambda)=g_{\cQ}(\lambda)$. Then $\cP\cup \cQ$ contains no $\{s,t\}$-uncrossable pair of intersecting sets.
    \end{lem}
    \begin{proof}
        For a contradiction, suppose $\cP\cup\cQ$ contains an $\{s,t\}$-uncrossable pair of intersecting sets. As long as such a pair $A,B\in\cP\cup\cQ$ exists, replace it with $A\cap B$ and $A\cup B$. By strict submodularity, each uncrossing strictly decreases the sum of the function values, so this procedure terminates in finitely many steps. Let $\cS$ denote the resulting family. 
        
        By the indirect assumption, we have $f(\cS)<f(\cP)+f(\cQ)$, while $|\cS|=|\cP|+|\cQ|$ and each element remains in exactly two sets. We claim that $\cS$ is the disjoint union of two $\{s,t\}$-separating partitions. 

        \begin{claim}\label{cl:st-sep}
            None of the sets in $\cS$ contains both $s$ and $t$.
        \end{claim}
        \begin{proof} 
        We show that each uncrossing step preserves the property that no set contains both $s$ and $t$. Initially, both $\cP$ and $\cQ$ are $\{s,t\}$-separating, so the claim holds. For the sake of contradiction, suppose the property holds before uncrossing a pair of intersecting sets $A,B$ but fails after the uncrossing. 
        
        If $\{s,t\}\subseteq A\cap B$, then $\{s,t\}\subseteq A$, contradicting the property before uncrossing. Therefore, we must have $\{s,t\}\subseteq A\cup B$. Since neither $A$ nor $B$ contains both $s$ and $t$, it follows that $|\{s,t\}\cap A|=1=|\{s,t\}\cap B|$, which would imply that $A$ and $B$ were not $\{s,t\}$-uncrossable. But then the procedure would not have uncrossed $A$ and $B$, a contradiction.
        \end{proof}

        Now we show that $\cS$ decomposes into two partitions. 
        
        \begin{claim}\label{cl:decomposes}
            $\cS$ is the disjoint union of two partitions.
        \end{claim}
        \begin{proof} 
        By Lemma~\ref{lem:no_uncross}, $\cS$ contains at most one pair of intersecting sets $X$ and $Y$ that are not $\{s,t\}$-uncrossable. If there is such a pair, then let $U \coloneqq  X\cup Y,$ otherwise let $U\coloneqq \emptyset$ for the rest of the proof of the claim. 
        Let 
        \begin{align*}
            \cP_1 &\coloneq \{X\} \cup \{A\in \cS \colon A\subsetneq Y\}\\
            \cQ_1 &\coloneq \{Y\} \cup \{B\in \cS \colon B\subsetneq X\}.
        \end{align*}
        Observe that $\cP_1$ and $\cQ_1$ are partitions of $U$, since $X$ and $Y$ form the only intersecting pair in $\cS$. The remainder, $\mathcal{L}\coloneqq\cS\setminus (\cP_1\cup\cQ_1)$, is a laminar system covering each element exactly twice, as it contains no intersecting sets. Let $\cP_2$ and $\cQ_2$ be the inclusion-wise maximal and minimal sets in $\mathcal{L}$, respectively, with any duplicates distributed between them. Then $\cP_2$ and $\cQ_2$ are partitions of $V\setminus U$. Setting $\cP'\coloneqq\cP_1\cup\cP_2$ and $\cQ'\coloneqq\cQ_1\cup\cQ_2$ thus yields two $\{s,t\}$-separating partitions of $V$ with $\cS=\cP'\cup\cQ'$.
        \end{proof}

        Let $\cP'$ and $\cQ'$ be the partitions provided by Claim~\ref{cl:decomposes}. By Claim~\ref{cl:st-sep}, $\cP'$ and $\cQ'$ are $\{s,t\}$-separating partitions. Recall that $f(\mcp')+f(\mcq')=f(\cS)<f(\mcp)+f(\mcq)$ and $|\mcp'|+|\mcq'|=|\mcs|=|\mcp|+|\mcq|$. Moreover, $g^{s,t}(\lambda) = g_\cP(\lambda) = g_\cQ(\lambda) $ by assumption and $g^{s,t}(\lambda) \leq g_{\cP'}(\lambda)$ and $g^{s,t}(\lambda) \leq g_{\cQ'}(\lambda)$ by definition. Combining these observations, we obtain
        \begin{align*}
            2g^{s,t}(\lambda) &= g_\cP(\lambda) + g_{\cQ}(\lambda)\\
            &= f(\cP) + f(\cQ) - \lambda (|\cP| + |\cQ|)\\
            &> f(\cP') + f(\cQ') - \lambda (|\cP'| + |\cQ'|)\\
            &=g_{\mcp'}(\lambda) + g_{\mcq'}(\lambda)\\
            &\geq 2g^{s,t}(\lambda),
        \end{align*}
        a contradiction. This concludes the proof of the lemma.
    \end{proof}

    The next lemma shows that under the hypothesis of Theorem~\ref{thm:refinement}, the parts of the partitions $\mcp$ and $\mcq$ cannot be rearranged to create two other $\{s,t\}$-separating partitions $\mcp'$ and $\mcq'$ with $|\mcq'|>|\mcq|$. The proof of the lemma in fact does not rely on submodularity of the set function $f$.
    
    \begin{lem}\label{lem:helpful}
        Suppose that $\mcp$ and $\mcq$ are distinct $\{s,t\}$-separating partitions such that one of the following holds: 
        \begin{enumerate}\itemsep0em
            \item There exists $\lambda\in \R$ such that 
            \begin{align*}
                \cP &\in \argmin \{|\pi|\colon \pi\text{ is an }\{s,t\}\text{-separating partition with }g_\pi(\lambda) = g^{s,t}(\lambda) \},\\
                \cQ &\in \argmax \{|\pi|\colon \pi\text{ is an }\{s,t\}\text{-separating partition with }g_\pi(\lambda) = g^{s,t}(\lambda) \}.
            \end{align*}\label{helpful:1}
            \item There exists $\alpha < \lambda $ such that $g^{s,t}(\alpha)=g_{\mcp}(\alpha)$ and $g^{s,t}(\lambda) = g_{\mcp}(\lambda) = g_{\mcq}(\lambda)$.\label{helpful:2} 
        \end{enumerate}
        Then no $\{s,t\}$-separating partitions $\cP'$ and $\cQ'$ exist with $\cP' \cup \cQ' = \cP \cup \cQ$ and $|\cQ'| > |\cQ|$.
    \end{lem}
    \begin{proof}
    For the sake of contradiction, suppose such partitions $\mcp'$ and $\mcq'$ exist. Considering the two assumptions in the lemma, we arrive at a contradiction in both cases.

    Suppose first that \ref{helpful:1} holds. Then, 
    \begin{align*}
        2g^{s,t}(\lambda) 
        &= g_{\mcp}(\lambda) + g_{\mcq}(\lambda) \\
        &= f(\mcp)+f(\mcq) - \lambda(|\mcp|+|\mcq|)\\
        &= f(\mcp') + f(\mcq') - \lambda(|\mcp'|+|\mcq'|)\\
        & = g_{\mcp'}(\lambda) + g_{\mcq'}(\lambda)\\
        &\ge 2g^{s,t}(\lambda).
    \end{align*}
    Consequently, we have $g_{\mcp'}(\lambda) = g_{\mcq'}(\lambda) = g^{s,t}(\lambda)$, contradicting the choice of $\mcq$.
    
    Now suppose that \ref{helpful:2} holds. Then, 
        \begin{align*}
            g_{\cP}(\alpha) + g_{\cQ}(\lambda) 
            &= g_{\cP}(\alpha) + g_{\cQ}(\alpha) - (\lambda - \alpha)|\cQ|\\
            &= f(\cP) + f(\cQ) - \alpha|\cP| - \alpha|\cQ| - (\lambda - \alpha)|\cQ|\\
            &= f(\cP') + f(\cQ') - \alpha|\cP| - \alpha|\cQ| - (\lambda - \alpha)|\cQ|\\
            &= g_{\cP'}(\alpha) + g_{\cQ'}(\alpha) - (\lambda - \alpha)|\cQ|\\
            &= g_{\cP'}(\alpha) + g_{\cQ'}(\lambda) + (\lambda - \alpha)(|\cQ'| - |\cQ|)\\
            &\geq g_{\cP}(\alpha) + g_{\cQ'}(\lambda) + (\lambda - \alpha)(|\cQ'| - |\cQ|).
        \end{align*}
        Hence, 
        \[ g_\cQ(\lambda) \geq g_{\cQ'}(\lambda) + (\lambda - \alpha)(|\cQ'| - |\cQ|). \]
        Since $\lambda>\alpha$ and $|\cQ'|>|\cQ|$, we have $g_{\cQ}(\lambda)>g_{\cQ'}(\lambda)\ge g^{s,t}(\lambda)=g_{\cQ}(\lambda)$, a contradiction.
    \end{proof}

    Now we are ready to prove Theorem~\ref{thm:refinement}. 
    
    \begin{proof}[Proof of Theorem~\ref{thm:refinement}]
    We consider two cases. 
    \medskip
    
    \noindent \textbf{Case 1.} There is no pair of intersecting sets in $\cP \cup \cQ$.
    
    Note that $\cP \cup \cQ$ is laminar. We show that $\mcq$ is a refinement of $\mcp$.
        For the sake of contradiction, suppose $\cQ$ is not a refinement of $\cP$. Then there exists $B\in \cQ$ such that $B\not\subseteq  A$ for all $A\in \cP$. Since $\cP\cup \cQ$ is laminar, at least two sets of $\cP$ are contained in $B$. Let $\cP' \coloneq \{B\} \cup {A\in \cP \colon A\cap B = \emptyset}$ and $\cQ' \coloneqq {A\in \cP \colon A\subseteq B} \cup (\cQ \setminus {B})$. Then $\cP'$ and $\cQ'$ are $\{s,t\}$-separating partitions with $\cP'\cup \cQ' = \cP\cup \cQ$ and $|\cQ'| > |\cQ|$, contradicting Lemma~\ref{lem:helpful}.
    \medskip

    \noindent \textbf{Case 2.} There exists a pair of intersecting sets in $\cP\cup \cQ$. 
    
    By Lemmas~\ref{lem:no_uncross} and~\ref{lem:key_lem}, there is exactly one intersecting pair $X,Y\in \cP\cup \cQ$. Since both $\cP$ and $\cQ$ are partitions, $X$ and $Y$ must lie in distinct partitions; thus, without loss of generality, assume that $X\in \cP\setminus \cQ$ and $Y\in \cQ\setminus \cP$. We complete the proof by showing that $\mcq$ is an $\{s,t\}$-refinement of $\mcp$ along $(X, Y)$. Since $X$ and $Y$ are not $\{s,t\}$-uncrossable, we have that $|\{s,t\}\cap (X\setminus Y)| =1= |\{s,t\}\cap (Y\setminus X)|$. The next three claims complete the proof of the remaining properties. 

    \begin{claim}\label{cl:b}
        Every set $B$ in $\cQ\setminus Y$ is contained in some set $A$ of $\cP$.
    \end{claim}
    \begin{proof}
        For the sake of contradiction, let $B \in \cQ \setminus \{Y\}$ be such that $B \not\subseteq A$ for all $A \in \cP$. The set $B$ cannot be intersecting with any of the sets in $\mcp \cup \mcq$, since $X$ and $Y$ are the only intersecting pair in $\mcp \cup \mcq$. Hence, for every $A \in \mcp$, either $B \cap A = \emptyset$ or $A \subseteq B$. Moreover, since $B \setminus A \neq \emptyset$ for all $A \in \mcp$, it follows that there exist at least two sets of $\cP$ that are contained in $B$.
            
        We now claim that $B\subseteq V\setminus(X\cup Y)$. To prove this, recall that $B\cap Y=\emptyset$ since $B$ and $Y$ are parts of the partition $\mcq$ and hence disjoint. It remains to show that $X\cap B=\emptyset$. Suppose, for the sake of contradiction, that $X\cap B\neq\emptyset$. Then, since for every $A\in\mcp$ we have either $B\cap A=\emptyset$ or $A\subseteq B$, it follows that $X\subseteq B$. Consequently, $B\cap Y\supseteq X\cap Y\neq\emptyset$, a contradiction.

        Consider the family $\mathcal{L} \coloneq \{S\in \cP \cup \cQ \colon S\cap X = \emptyset = S\cap Y \}$. Since $X$ and $Y$ are the only intersecting sets in $\mcp\cup \mcq$, the family $\mathcal{L}$ is laminar. Each element of $V\setminus(X\cup Y)$ is contained in exactly two sets of $\mathcal{L}$. Let $\cP_1$ and $\cQ_1$ be the inclusion-wise maximal and minimal sets in $\mathcal{L}$, respectively, with any duplicates distributed between the two. Let 
        \begin{align*}
            \cP' &\coloneq \{X\} \cup \{A\in \cP \colon A\subseteq Y\} \cup \cP_1\\
            \cQ' &\coloneq \{Y\} \cup \{C\in \cQ \colon C\subseteq X\} \cup \cQ_1.
        \end{align*}
        Then both $\cP'$ and $\cQ'$ are $\{s,t\}$-separating partitions of $V$ with $\cP'\cup \cQ' = \cP\cup \cQ$. Since $B\in \cP_1$ and there are at least two sets of $\mathcal{L}$ contained in $B$, we get $|\cQ'| > |\cQ|$, contradicting Lemma~\ref{lem:helpful}.
    \end{proof}

    \begin{claim}\label{claim:P-X-laminar}
        Every set $A$ in $\cP \setminus X$ is either disjoint from $Y$ or contained in $Y$.
    \end{claim}
    \begin{proof}
        For the sake of contradiction, let $A\in \cP\setminus \{X\}$ satisfy $A\cap Y \neq \emptyset$ and $A\setminus Y \neq \emptyset$. The sets $Y$ and $A$ cannot be intersecting, since $X$ and $Y$ are the only intersecting pair in $\cP\cup \cQ$. Thus $Y\setminus A = \emptyset$, which implies $Y\subseteq A$. Then $A\cap X \supseteq Y\cap X \neq \emptyset$, contradicting the fact that $A$ and $X$ are both in the partition $\cP$ and therefore disjoint.
    \end{proof}


    \begin{claim}\label{cl:d}
        $|\{A\in \cP \colon A\subseteq Y\}| \leq |\{B\in \cQ \colon B\subseteq X\}|$.
    \end{claim}
    \begin{proof}
    Suppose indirectly that the claim fails. Consider 
        \begin{align*}
            \cP' &\coloneq \{B\in \cQ \colon B\subseteq X\cup Y\} \cup \{A\in \cP \colon A\subseteq V\setminus(X\cup Y)\}\\
            \cQ' &\coloneq \{A\in \cP \colon A\subseteq X\cup Y\} \cup \{B\in \cQ \colon A\subseteq V\setminus (X\cup Y)\}.
        \end{align*}
        Then $\cP'$ and $\cQ'$ are $\{s,t\}$-separating partitions with $\mcp'\cup \mcq' = \mcp\cup \mcq$. By the indirect assumption, we get $|\cQ'| > |\cQ|$, contradicting Lemma~\ref{lem:helpful}.
    \end{proof}

    Thus we get that properties \ref{strefinement:1b}, \ref{strefinement:1c} and \ref{strefinement:1d} of Definition~\ref{def:strefine}\eqref{strefinement:1} hold by Claims~\ref{cl:b}, \ref{claim:P-X-laminar} and \ref{cl:d}, respectively. That is,  $\mcq$ is an $\{s,t\}$-refinement of $\mcp$ along $(X, Y)$, concluding the proof of the theorem. 
\end{proof}

\subsection{Existence}
\label{subsec:st-pps-existence}

The goal of this section is to show that submodular function that is strictly submodular on intersecting pairs admits an $\{s,t\}$-separating principal partition sequence. 

\begin{thm}\label{thm:stppsexists}
    Let $f\colon 2^V\to \bR$ be a submodular function 
    that is strictly submodular on intersecting pairs,
    and let $s,t\in V$. Then, there exists an $\{s,t\}$-separating principal partition sequence of $f$. 
\end{thm}

The proof of Theorem~\ref{thm:stppsexists} relies on three lemmas that characterize how the partitions minimizing $g^{s,t}(\lambda)$ behave between breakpoints and at a breakpoint.

\begin{lem}\label{lem:partitions-between-breakpoints}
    Let $\lambda_a<\lambda_{b}$ be adjacent breakpoints of $g^{s,t}$ and let 
    \begin{align*}
        \cP_a &\in \argmax \{|\pi|\colon \pi\text{ is an }\{s,t\}\text{-separating partition with }g_\pi(\lambda_a) = g^{s,t}(\lambda_a) \},\\
        \cP_b &\in \argmin \{|\pi|\colon \pi\text{ is an }\{s,t\}\text{-separating partition with }g_\pi(\lambda_b) = g^{s,t}(\lambda_b) \}.
    \end{align*}
    Then, $g_{\cP_a}(\lambda)=g^{s,t}(\lambda)=g_{\cP_b}(\lambda)$ for all $\lambda\in[\lambda_a,\lambda_b]$, and $|\cP_a|=|\cP_b|$.

    Moreover, if $g_{\cP}(\lambda)=g^{s,t}(\lambda)$ for some $\lambda\in (\lambda_a,\lambda_b)$ for a partition $\mcp$, then 
    $g_{\cP}(\lambda)=g^{s,t}(\lambda)$ for every $\lambda\in [\lambda_a, \lambda_b]$, and $|\mcp| = |\mcp_a| = |\mcp_b|$.
    
\end{lem}
\begin{proof}
    Since $\lambda_a$ and $\lambda_b$ are adjacent breakpoints of $g^{s,t}$, there exists an $\{s,t\}$-separating partition $\mcr$ such that $g_{\mcr}(\lambda)=g^{s,t}(\lambda)$ for every $\lambda\in [\lambda_a, \lambda_b]$. 
    
    We first show that $g_{\mcp_a}(\lambda)=g^{s,t}(\lambda)$ for every $\lambda \in [\lambda_a, \lambda_b]$. By the definition of $\mcp_a$, it follows that $|\mcr|\le|\mcp_a|$. Applying Theorem~\ref{thm:refinement} with $\mcp=\mcp_a$, $\mcq=\mcr$, $\alpha=\lambda_a$, and $\lambda=\lambda_b$, we obtain that $\mcr$ is either a refinement or an $\{s,t\}$-refinement of $\mcp_a$, and consequently $|\mcr|\ge|\mcp_a|$. Thus $|\mcr|=|\mcp_a|$. Since $f(\mcr)-\lambda_a|\mcr|=f(\mcp_a)-\lambda_a|\mcp_a|$ and $|\mcr|=|\mcp_a|$, it follows that $f(\mcr)=f(\mcp_a)$. Consequently, $g_{\mcp_a}(\lambda)=g_{\mcr}(\lambda)$ for every $\lambda$, and in particular $g_{\mcp_a}(\lambda)=g^{s,t}(\lambda)$ for every $\lambda\in[\lambda_a,\lambda_b]$.


   Next, we show that $g_{\mcp_b}(\lambda)=g^{s,t}(\lambda)$ for every $\lambda\in[\lambda_a,\lambda_b]$. By the definition of $\mcp_b$, it follows that $|\mcr|\ge|\mcp_b|$. Applying Theorem~\ref{thm:refinement} with $\mcp=\mcr$, $\mcq=\mcp_b$, $\alpha=\lambda_a$, and $\lambda=\lambda_b$, we obtain that $\mcp_b$ is either a refinement or an $\{s,t\}$-refinement of $\mcr$, and consequently $|\mcp_b|\ge|\mcr|$. Thus $|\mcr|=|\mcp_b|$. Since $f(\mcr)-\lambda_b|\mcr|=f(\mcp_b)-\lambda_b|\mcp_b|$ and $|\mcr|=|\mcp_b|$, it follows that $f(\mcr)=f(\mcp_b)$. Consequently, $g_{\mcp_b}(\lambda)=g_{\mcr}(\lambda)$ for every $\lambda$, and in particular $g_{\mcp_b}(\lambda)=g^{s,t}(\lambda)$ for every $\lambda\in[\lambda_a,\lambda_b]$.

    Let $\cP_c$ be such that $g_{\cP_c}(\lambda_c)=g^{s,t}(\lambda_c)$ for some $\lambda_c\in (\lambda_a,\lambda_b)$. Applying Theorem~\ref{thm:refinement} with $\mcp=\mcp_c$, $\mcq=\mcp_a$, $\alpha=\lambda_c$, and $\lambda=\lambda_a$, we obtain that $\mcp_a$ is either a refinement or an $\{s,t\}$-refinement of $\mcp_c$, and consequently $|\mcp_a|\ge|\mcp_c|$. Similarly, applying Theorem~\ref{thm:refinement} with $\mcp_c=\mcp_b$, $\mcq=\mcp_c$, $\alpha=\lambda_b$, and $\lambda=\lambda_c$, we obtain that $\mcp_c$ is either a refinement or an $\{s,t\}$-refinement of $\mcp_b$, and consequently $|\mcp_c|\ge|\mcp_b|$. Thus $|\mcp| = |\mcp_a| = |\mcp_b|.$ Since $g_{\cP_a}(\lambda_c)=g_{\cP_c}(\lambda_c)=g_{\cP_b}(\lambda_c)$, it follows that $f{(\cP_a)} = f(\cP_c) = f({\cP_b})$. Consequently, $g_{\mcp_b}(\lambda) = g_{\mcp_c}(\lambda)=g_{\mcp_b}(\lambda)$ for every $\lambda\in[\lambda_a,\lambda_b]$.
    
\end{proof}

The following lemma characterizes the behavior of the partitions minimizing $g^{s,t}(\lambda)$ at the leftmost and rightmost breakpoints. The proof does not rely on submodularity and holds for arbitrary set functions.

\begin{lem}\label{lem:partitions-at-left-and-right-breakpoints}
    Let $a, b$ be the leftmost and rightmost breakpoints of $g^{s,t}$ and let 
    \begin{align*}
        \mcp'&\in\argmin\{f(\pi)\colon \pi\text{ is an }\{s,t\}\text{-separating partition with } |\pi|=2\}, \\
        \mcq'&\coloneqq \{\{v\}\colon v\in V\}. 
    \end{align*}
    Then, $g^{s,t}(\lambda) = g_{\cP'}(\lambda)$ for $\lambda\in (-\infty, a]$ and $g^{s,t}(\lambda)=g_{\mcq'}(\lambda)$ for $\lambda\in [b, \infty)$. 
\end{lem}
\begin{proof}

We first show that $g^{s,t}(\lambda)=g_{\cP'}(\lambda)$ for all $\lambda\in(-\infty,a]$. Since $a$ is the leftmost breakpoint of $g^{s,t}$, there exists an $\{s,t\}$-separating partition $\mcr$ such that $g^{s,t}(\lambda)=g_{\mcr}(\lambda)$ for every $\lambda\in(-\infty,a]$. Assume that $|\mcr|\ge3$. Let $\lambda<\min{a,(f(\mcr)-f(\mcp'))/(|\mcr|-2)}$. For such $\lambda$, we have
\begin{align*}
    g^{s,t}(\lambda)
    \le g_{\mcp'}(\lambda)
    = f(\mcp') - 2\lambda
    < f(\mcr) - \lambda|\mcr|
    =g_{\mcr}(\lambda) 
    =g^{s,t}(\lambda),
\end{align*}
a contradiction. Thus we conclude that $|\mcr|=2$, yielding $g_{\mcp'}(\lambda)=g_{\mcr}(\lambda)$ for every $\lambda\in(-\infty,a]$. 

Next, we show that $g^{s,t}(\lambda)=g_{\cQ'}(\lambda)$ for all $\lambda\in[b,\infty)$. Since $b$ is the rightmost breakpoint of $g^{s,t}$, there exists an $\{s,t\}$-separating partition $\mcr$ such that $g^{s,t}(\lambda)=g_{\mcr}(\lambda)$ for every $\lambda\in[b,\infty)$. Assume that $|\mcr|<|V|$. Let $\lambda>\max{b,(f(\mcq')-f(\mcr))/(|V|-|\mcr|)}$. For such $\lambda$, we have
\begin{align*}
    g^{s,t}(\lambda)
    \le g_{\mcq'}(\lambda)
    = f(\mcq') - \lambda|V|
    < f(\mcr) - \lambda|\mcr|
    =g_{\mcr}(\lambda) 
    =g^{s,t}(\lambda),
\end{align*}
a contradiction. Thus we conclude that $|\mcr|=|V|$, yielding $\mcr =\mcq'$. 
\end{proof}

\begin{lem}\label{lem:partitions-at-breakpoint}
    Let $\lambda$ be a breakpoint of $g^{s,t}$ and let 
    \begin{align*}
        \mcp &\in \argmin \{|\pi|\colon \pi\text{ is an }\{s,t\}\text{-separating partition with }g_\pi(\lambda) = g^{s,t}(\lambda) \}.
    \end{align*}
    Then, there exists 
    a sequence $\cP = \mcr_0,\mcr_1, \ldots, \mcr_{r+1}$ of $\{s,t\}$-separating partitions such that
    \begin{enumerate}\itemsep0em
        \item $\mcr_{r+1}\in \argmax \{|\pi|\colon \pi\text{ is an }\{s,t\}\text{-separating partition with }g_\pi(\lambda) = g^{s,t}(\lambda) \}$, 
        \item $\mcr_1$ is either a refinement of $\cP$ up to one set or an $\{s,t\}$-refinement of $\cP$ up to two sets, 
        \item $\mcr_{i+1}$ is a refinement of $\mcr_i$ up to one set for every $i\in [r]$, 
        \item $|\mcr_{i+1}|>|\mcr_i|$ for every $i\in \{0,1,2,\ldots, r\}$, and 
        \item $g_{\mcr_i}(\lambda) = g^{s,t}(\lambda)$ for every $i\in \{0, 1, \ldots, r+1\}$. 
    \end{enumerate}
    Moreover, given 
    \begin{align*}
        \mcp &\in \argmin \{|\pi|\colon \pi\text{ is an }\{s,t\}\text{-separating partition with }g_\pi(\lambda) = g^{s,t}(\lambda) \}, \text{ and }\\
        \mcq &\in \argmax \{|\pi|\colon \pi\text{ is an }\{s,t\}\text{-separating partition with }g_\pi(\lambda) = g^{s,t}(\lambda) \}
    \end{align*}
    a sequence $\cP = \mcr_0,\mcr_1, \ldots, \mcr_{r+1} = \mcq$ of $\{s,t\}$-separating partitions satisfying the above five properties can be constructed in polynomial-time.
\end{lem}

\begin{proof}
    We prove the existence. Our proof of existence is also constructive and it leads to the polynomial-time algorithm mentioned in the lemma. 
    Let $\mcq'\in \argmax \{|\pi|\colon \pi\text{ is an }\{s,t\}\text{-separating partition with }g_\pi(\lambda) = g^{s,t}(\lambda) \}$. 
    By Theorem~\ref{thm:refinement}, $\mcq'$ is either a refinement or an $\{s,t\}$-refinement of $\mcp$. Let $U\coloneqq \emptyset$ if $\mcq'$ is a refinement of $\mcp$ and $U\coloneqq X\cup Y$ if $\mcq'$ is an $\{s,t\}$-refinement of $\mcp$ along $(X, Y)$. It then follows that $|\{A\in \mcp\colon A\subseteq V\setminus U\}|\le|\{B\in \mcq'\colon B\subseteq V\setminus U\}|$.
    
    First, consider the case when $|\{A\in \mcp\colon A\subseteq V\setminus U\}|=|\{B\in \mcq'\colon B\subseteq V\setminus U\}|$. Then, we have $\{A\in \mcp\colon A\subseteq V\setminus U\}=\{B\in \mcq'\colon B\subseteq V\setminus U\}$. Consequently, $\mcq'$ is a refinement of $\mcp$ up to one set or an $\{s,t\}$-refinement of $\mcp$ up to two sets. 
    Since $\lambda$ is a breakpoint of $g^{s,t}$, it follows that the slope of the function $g_{\mcp}$ at $\lambda$ is strictly more than the slope of $g_{\mcq'}$ at $\lambda$ and consequently, $|\mcp|<|\mcq'|$. If $\mcq'$ is a refinement of $\mcp$, then $U=\emptyset$ and consequently, $\{A\in \mcp\colon A\subseteq V\setminus U\}=\{B\in \mcq'\colon B\subseteq V\setminus U\}$ implies that $\mcp=\mcq'$, thus contradicting $|\mcp|<|\mcq'|$. Hence, $\mcq'$ is an $\{s,t\}$-refinement of $\mcp$ up to two sets. 
    Therefore, setting $r=0$, $\mcr_0=\mcp$, and $\mcr_1=\mcq'$ satisfies the requirements of the lemma.

    Now consider the case when $|\{A\in \mcp\colon A\subseteq V\setminus U\}|<|\{B\in \mcq'\colon B\subseteq V\setminus U\}|$. We know that for every $B\in \mcq'$ there exists $A\in \mcp$ with $B\subseteq A$. By the assumption of the case, there exists $A\in \mcp$ with $A\subseteq V\setminus U$ that contains at least two parts of $\mcq'$. Let $P_1, \dots, P_{r+1}$ be such parts of $\mcp$, that is, $P_i\subseteq V\setminus U$ and each $P_i$ contains at least two parts of $\mcq'$ for every $i\in [r+1]$. For each $i\in [r+1]$, we define
    \begin{align*}
        \mcr_i &\coloneqq  \{B\in \mcq'\colon B\subseteq U\}\cup \bigl\{B\in \mcq'\colon B\subseteq \bigcup_{j=1}^i P_j\bigr\}\cup \bigl\{A\in \mcp\colon A\subseteq (V\setminus U)\setminus  \bigcup_{j=1}^i P_j\bigr\},\\
        \mcs_i &\coloneqq  \{A\in \mcp\colon A\subseteq U\}\cup \bigl\{A\in \mcp\colon A\subseteq \bigcup_{j=1}^i P_j\bigr\}\cup \bigl\{B\in \mcq'\colon B\subseteq (V\setminus U)\setminus  \bigcup_{j=1}^i P_j\bigr\}. 
    \end{align*}
    Moreover, define 
    \begin{align*}
        \mcr_0 &\coloneqq  \{B\in \mcq'\colon B\subseteq U\}\cup \left\{A\in \mcp\colon A\subseteq V\setminus U\right\}\\
        \mcs_0 &\coloneqq  \{A\in \mcp\colon A\subseteq U\}\cup \left\{B\in \mcq'\colon B\subseteq V\setminus U\right\}. 
    \end{align*}
    
    Observe that $\mcr_{r+1}=\mcq'$, $\mcs_{r+1}=\mcp$, $\mcr_i$ is a refinement of $\mcr_{i-1}$ up to one set for every $i\in[r+1]$, and if $U=\emptyset$, then $\mcr_0=\mcp$, whereas if $U\neq\emptyset$, then $\mcr_0$ is an $\{s,t\}$-refinement of $\mcp$ up to two sets. 
    Since $\mcr_i$ is a strict refinement of $\mcr_{i-1}$ up to one set for every $i\in [r+1]$, we immediately have that $|\mcr_i|> |\mcr_{i-1}|$ for every $i\in [r+1]$. Similarly, $|\mcs_{i-1}| > |\mcs_{i}|$ for every $i\in [r+1]$.
    
    We claim that $g_{\mcr_i}(\lambda) = g^{s,t}(\lambda)=g_{\mcs_i}(\lambda)$ for every $i\in \{0, 1, \ldots, r\}$. To see this, note that $\mcr_i$ and $\mcs_i$ are $\{s,t\}$-separating partitions of $V$ with $\mcr_i \cup \mcs_i = \mcp\cup \mcq'$, yielding $g^{s,t}(\lambda ) \le g_{\mcr_i}(\lambda)$ and $g^{s,t}(\lambda) \le g_{\mcs_i}(\lambda)$ for $i\in\{0,1,\dots,r\}$. Hence 
        \begin{align*}
        2g^{s,t}(\lambda) 
        &\le g_{\mcr_i}(\lambda) + g_{\mcs_i}(\lambda) \\
        &= f(\mcr_i) + f(\mcs_i)- \lambda(|\mcr_i| + |\mcs_i|) \\
        &= f(\mcp) + f(\mcq') - \lambda(|\mcp|+|\mcq'|)\\
        &= g_{\mcp}(\lambda) + g_{\mcq'}(\lambda)\\
        &=2g^{s,t}(\lambda), 
        \end{align*}
    implying $g_{\mcr_i}(\lambda) = g^{s,t}(\lambda)=g_{\mcs_i}(\lambda)$. 
    
    Concluding the above, if $U=\emptyset$, then the sequence $\mcr_0, \mcr_1, \dots, \mcr_{r+1}$ satisfies the properties of the lemma. We now consider the case where $U\neq \emptyset$. 
    If $|\{A\in \mcp: A\subseteq U\}|<|\{B\in \mcq': B\subseteq U\}|$, then the sequence $\mcp, \mcr_0, \mcr_1, \dots, \mcr_{r+1}$ satisfies the properties of the lemma.
    Suppose $|\{A\in \mcp: A\subseteq U\}|=|\{B\in \mcq': B\subseteq U\}|$. Then, we observe that $|\mcs_0|=|\mcq'|$ and moreover, $g^{s,t}(\lambda)=g_{\mcs_0}(\lambda)$. Consequently, $\mcs_0\in \argmax \{|\pi|\colon \pi\text{ is an }\{s,t\}\text{-separating partition with }g_\pi(\lambda) = g^{s,t}(\lambda) \}$. 
    Hence, the sequence $\mcs_{r+1}, \mcs_r, \ldots, \mcs_0$ satisfies the properties of the lemma. 
\end{proof}

We now prove Theorem~\ref{thm:stppsexists}. 

\begin{proof}[Proof of Theorem~\ref{thm:stppsexists}]
    Let $\lambda_1 <\dots < \lambda_{\ell-1}$ be the breakpoints of $g^{s,t}$. We will use Lemmas \ref{lem:partitions-at-breakpoint} and \ref{lem:partitions-between-breakpoints} alternatively to construct the claimed sequence of partitions along with the critical value sequence. Let 
    \begin{align*}
        \mcp_1 &\in \argmin \{|\pi|\colon \pi\text{ is an $\{s,t\}$-separating partition with } g_{\pi}(\lambda_1) = g^{s,t}(\lambda_1)\}. 
    \end{align*}
    By Lemma~\ref{lem:partitions-at-left-and-right-breakpoints}, we have $\mcp_1 \in\argmin\{f(\mcp)\colon \mcp \text{ is an }\{s, t\}\text{-separating partition with }|\mcp|=2 \}$ and $g^{s,t}(\lambda) = g_{\mcp_1'}(\lambda)$ if $\lambda\in (-\infty, \lambda_1]$. 

    We proceed as follows for $i=1, 2, 3, \ldots, \ell-1$. 
    We will inductively ensure that the $\{s,t\}$-separating partition $\mcp_i$ is such that $\mcp_i\in \argmin\{|\pi|\colon \pi\text{ is an $\{s,t\}$-separating partition with } g_{\pi}(\lambda_i) = g^{s,t}(\lambda_i)\}$. 
    Applying Lemma \ref{lem:partitions-at-breakpoint} at the breakpoint $\lambda_i$ and the $\{s,t\}$-separating partition $\mcp_i$ gives a sequence of $\{s,t\}$-separating partitions $\mcp_i=\mcr^i_0, \mcr^i_1, \ldots, \mcr^i_{r_i+1}$ satisfying the properties guaranteed by the lemma. 
    We set $\mcp_{i+1}:=\mcr^i_{r_i+1}$. 
    By the first property of Lemma \ref{lem:partitions-at-breakpoint}, we know that 
    \[
    \mcp_{i+1}=\mcr^i_{r_i+1}\in \argmax \{|\pi|\colon \pi\text{ is an $\{s,t\}$-separating partition with } g_{\pi}(\lambda_i) = g^{s,t}(\lambda_i)\}. 
    \]
    Applying Lemma \ref{lem:partitions-between-breakpoints} to adjacent breakpoints $\lambda_i<\lambda_{i+1}$ of $g^{s,t}$, we 
    conclude that 
    \[
    \mcp_{i+1}\in \argmin \{|\pi|\colon \pi\text{ is an $\{s,t\}$-separating partition with } g_{\pi}
    (\lambda_{i+1}) = g^{s,t}(\lambda_{i+1})\}. 
    \]
    This implies that $\mcp_{i+1}$ satisfies the induction hypothesis for every $i\in [\ell-2]$ and moreover, $\mcp_{\ell}=\{\{v\}: v\in V\}$ by Lemma \ref{lem:partitions-at-left-and-right-breakpoints}. 
    
    With this construction, we now consider the sequence
    \[
    \mcp_1, \mcr^1_1, \mcr^1_2, \ldots, \mcr^1_{r_1+1}=\mcp_2, \ldots, 
    \mcp_i, \mcr^i_1, \mcr^i_2, \ldots, \mcr^i_{r_i+1}=\mcp_{i+1},\ldots, 
    \mcp_{\ell-1}, \mcr^{\ell}_1, \mcr^{\ell}_2, \ldots, \mcr^{\ell}_{r_{\ell}+1}=\mcp_{\ell}
    \]
    together with the critical value sequence
    \[
    \lambda_1, \underbrace{\lambda_1, \ldots, \lambda_1}_{r_1\text{ copies}}, \lambda_2, \ldots, \lambda_i, \underbrace{\lambda_i, \ldots, \lambda_i}_{r_i\text{ copies}}, \lambda_{i+1}, \ldots, \lambda_{\ell-1}, \underbrace{\lambda_{\ell-1}, \ldots, \lambda_{\ell-1}}_{r_{\ell-1}\text{ copies}}.
    \]
    We claim that this sequence of partitions satisfies properties \ref{stpps:1}-\ref{stpps:4} of Definition~\ref{def:st-pps}. The critical value sequence is non-decreasing and hence, we have \ref{stpps:1}. We have already seen that $\mcp_1 \in\argmin\{f(\mcp)\colon \mcp \text{ is an }\allowbreak\{s, t\}\text{-separating partition with }|\mcp|=2 \}$ and $\mcp_{\ell}=\{\{v\}: v\in V\}$ and hence, we have \ref{stpps:3}. 
    We now show \ref{stpps:2}. By Lemma \ref{lem:partitions-at-left-and-right-breakpoints}, we have that $g^{s,t}(\lambda) = g_{\mcp_1}(\lambda)$ for $\lambda\in (-\infty, \lambda_1]$ and $g^{s,t}(\lambda)=g_{\mcp_{\ell}}(\lambda)$ for $\lambda\in [\lambda_{\ell-1}, \infty)$. 
    Let $i\in [\ell-1]$. We observe that $g^{s,t}(\lambda_i)=g_{\mcr^i_j}(\lambda_i)$ by construction of $\mcr^i_j$ via Lemma \ref{lem:partitions-at-breakpoint}. It suffices to show that $g^{s,t}(\lambda)=g_{\mcp_{i+1}}(\lambda)$ for every $\lambda\in [\lambda_i, \lambda_{i+1}]$. This follows by applying Lemma \ref{lem:partitions-between-breakpoints} by considering adjacent breakpoints $\lambda_i<\lambda_{i+1}$ of $g^{s,t}$ and the partition $\mcp_{i+1}\in \argmax \{|\pi|\colon \pi\text{ is an $\{s,t\}$-separating}\allowbreak\text{ partition with }\allowbreak g_{\pi}(\lambda_i) = g^{s,t}(\lambda_i)\}$. 
    Finally, we show \ref{stpps:4}. Let $i\in [\ell-1]$. We observe that $\mcr^i_{1}$ is either a refinement of $\mcp_i$ or an $\{s,t\}$-refinement of $\mcp_i$ up to two sets and $|\mcr^i_1|>|\mcp_i|$ by Lemma \ref{lem:partitions-at-breakpoint}. Moreover, for $j>1$, we also have that $\mcr^i_j$ is either a refinement of $\mcr^i_{j-1}$ or an $\{s,t\}$-refinement of $\mcr^i_{j-1}$ up to two sets and $|\mcr^i_j|>|\mcr^i_{j-1}|$ by Lemma \ref{lem:partitions-at-breakpoint}. 

\end{proof}

\subsection{Algorithm}
\label{subsec:st-pps-algo}

In this section, we present a polynomial-time algorithm to compute an $\{s,t\}$-separating principal partition sequence. Since $g^{s,t}$ is a piecewise-linear concave function with at most $|V|-2$ breakpoints, all of its breakpoints can be found in polynomial time using the Newton--Dinkelbach method, provided that the value of $g^{s,t}(\lambda)$ can be computed for any given $\lambda$; see e.g.~\cite{cunningham,pps}. We now show that $g^{s,t}(\lambda)$ is indeed computable in polynomial time. In particular, we can obtain a partition $\mcp\in\argmin\{g_{\pi}(\lambda)\colon \pi\text{ is an $\{s,t\}$-separating}\allowbreak \text{partition with }\allowbreak g_{\pi}(\lambda)=g^{s,t}(\lambda)\}$.

\begin{thm}\label{thm:min-g}
    Given a submodular function $f\colon 2^V\rightarrow \R$ via its valuation oracle and $\lambda\in \mathbb{R}$, there exist polynomial-time algorithms to find a partition in the following collections: 
    \begin{enumerate}\itemsep0em
    \item $\argmin\{f(\cP) - \lambda |\cP|\colon \cP\text{ is a partition of }V\}$, and \label{eq:algo-partition}
    \item $\argmin\{f(\cP) - \lambda |\cP|\colon \cP\text{ is an }\{s,t\}\text{-separating partition of }V\}$. \label{eq:algo-st-partition}
    \end{enumerate}
\end{thm}
\begin{proof}
    The ideas for both problems are similar. The problem in \ref{eq:algo-partition} was already known to be solvable in polynomial time -- see e.g.~\cite{pps, nagano2010minimum}. We show that the same approach also extends to solve the problem in \ref{eq:algo-st-partition}. Consider the function $f_{\lambda}\colon 2^V\rightarrow \R$ defined as $f_{\lambda}(U)\coloneq f(U)-\lambda$ if $\emptyset\neq U\subsetneq V$ and $f_{\lambda}(U)=0$ otherwise. Then, the function $f_{\lambda}$ is intersecting submodular -- see e.g.~\cite{nagano2010minimum}. 

    For a function $p\colon 2^V\rightarrow \R$, consider the function $p^{\lor}\colon 2^V\rightarrow \R$ defined as 
    \begin{align*}
        p^{\lor}(U) &\coloneq \min\{ p(\cP)\colon \cP \text{ is a partition of }U \}.
    \end{align*}
    The function $p^{\lor}$ is known as the Dilworth truncation of $p$. Suppose $p$ is intersecting submodular. Then $p^{\lor}$ is submodular. Moreover, given evaluation oracle access to $p$, there exists a polynomial-time algorithm to find a partition $\mcp$ of a given subset $U$ such that $p^{\lor}(U) = p(\mcp)$ 
    \cite{fujishige2005submodular}. For any submodular function $q\colon 2^V\to\R$, the complement function $\overline{q}\colon 2^V\to\R$ defined by $\overline{q}(U)\coloneq q(V\setminus U)$ for every $U\subseteq V$ is also submodular. Consequently, the function $m\coloneq \overline{f_{\lambda}^{\lor}}$ is submodular and admits a polynomial-time evaluation oracle via the oracle for $f$. We note that the function $m\colon 2^V\rightarrow \R$ is given by  
    \begin{align*} 
    m(U) & \coloneq \min\{ f(\cP) - \lambda |\cP| \colon \cP \text{ is a partition of } V\setminus U \}. 
    \end{align*}
    and there exists a polynomial time algorithm that takes a subset $U$ as input and finds a partition $\mcp$ of $V\setminus U$ such that $m(U) = f(\mcp)-\lambda|\mcp|$.
    %
    Next, consider the function $b\colon 2^V\rightarrow \R$ defined as 
    \[
    b(U) \coloneq f(U) -\lambda + m(U).
    \]
   
    \begin{claim}\label{cl:equiv}
    We have the following:
    \begin{enumerate}\itemsep0em
        \item $\min\{b(U)\colon s\in U\subseteq V-t\}=\min\{f(\cP) - \lambda |\cP|\colon \cP\text{ is an }\{s,t\}\text{-separating partition of }V\}$. \label{equiv:2}
        \item $\min\{b(U)\colon U\subseteq V\}=\min\{f(\cP) - \lambda |\cP|\colon \cP\text{ is a partition of }V\}$. \label{equiv:1}
    \end{enumerate}
    \end{claim}
    \begin{proof}
    We prove \ref{equiv:2}; the proof of \ref{equiv:1} is analogous. 
    Let $W \in \argmin\{b(U)\colon s\in U\subseteq V-t\}$ and $\cP_1\in \argmin\{ f(\cP) - \lambda |\cP| \colon \cP \text{ is a partition of } V\setminus W\}$. Consider $\cP_2 \coloneq  \{W\}\cup \cP_1$. Then, $\cP_2$ is an $\{s,t\}$-separating partition of $V$ such that $b(W) = f(W) - \lambda + f(\cP_1) - \lambda|\cP_1| = f(\cP_2) - \lambda|\cP_2|$, showing that the minimum on the left hand side is at least the minimum on the right hand side.
    
    For the reverse direction, let $\cP_2\in \argmin\{f(\pi) - \lambda |\pi|\colon \pi \text{ is an }\{s,t\}\text{-separating partition of $V$}\}$. Let $W$ be the part in $\cP_2$ containing $s$, implying $t\not\in W$. Consider $\cP_1 \coloneq  \cP_2\setminus \{W\}$. Then, $W$ is a subset of $V$ such that $b(W)=f(W)-\lambda+m(W)\le f(W)-\lambda + f(\mcp_1)-\lambda|\mcp_1|=f(\cP_2) - \lambda |\cP_2|$, finishing the proof of the claim. 
\end{proof}
Since $f$ and $m$ are submodular, it follows that the function $b$ is also submodular. Moreover, we have a polynomial-time evaluation oracle for $b$ using the evaluation oracle for $f$. 
    Thus, we can compute $U\in \argmin\{b(U)\colon s\in U\subseteq V-t\}$ (and similarly, $U\in \argmin\{b(U)\colon U\subseteq V\}$) via submodular minimization in polynomial time using the evaluation oracle for $b$. 
    Now, we recall that there exists a polynomial-time algorithm to compute a partition $\mcp'$ of $V\setminus U$ such that $m(U) = f(\mcp')-\lambda|\mcp'|$. Now, we return the partition $\mcp\coloneq \mcp'\cup\{U\}$. This partition is indeed in $\argmin\{f(\cP) - \lambda |\cP|\colon \cP\text{ is an }\{s,t\}\text{-separating partition of }V\}$ (and in $\argmin\{f(\cP) - \lambda |\cP|\colon \cP\text{ is a partition of }V\}$ respectively)  by Claim~\ref{cl:equiv}. 
 \end{proof}

We now complete the proof of Theorem \ref{lem:strict-st-pps}.

\begin{proof}[Proof of Theorem \ref{lem:strict-st-pps}]
The existence of the required sequence follows by Theorem \ref{thm:stppsexists}. We now give an algorithm to construct such a sequence. 
Using Theorem~\ref{thm:min-g} and the Newton-Dinkelbach method (e.g., see~\cite{cunningham, pps}), we can compute the breakpoints of $g^{s,t}$ in polynomial time. In order to compute an $\{s,t\}$-separating principal partition sequence from the breakpoints, we follow the proof of Theorem~\ref{thm:stppsexists} on the existence. In particular, at each breakpoint $\lambda_i$ of $g^{s,t}$, by the algorithmic conclusion of Lemma \ref{lem:partitions-at-breakpoint}, 
it suffices to find a partition in each of the following collections:
\begin{align}
    \argmin &\{|\pi|\colon \pi\text{ is an $\{s,t\}$-separating partition with } g_{\pi}(\lambda_i) = g^{s,t}(\lambda_i)\}, \label{eq:argmin-problem}\\
    \argmax&\{|\pi|\colon \pi \text{ is an $\{s,t\}$-separating partition with } g_{\pi}(\lambda_i) = g^{s,t}(\lambda_i)\}. \label{eq:argmax-problem}
\end{align}


This can be done as follows: For each $1<i<\ell-1$, use Theorem~\ref{thm:min-g} to find partitions $\cP_a$ and $\cP_b$ minimizing $g^{s,t}\left(\lambda_{a}\right)$ and $g^{s,t}\left(\lambda_{b}\right)$ where $\lambda_{a} = \frac{\lambda_{i-1}+\lambda_{i}}{2}$ and $\lambda_{b} = \frac{\lambda_{i}+\lambda_{i+1}}{2}$, respectively.
By Lemma \ref{lem:partitions-between-breakpoints}, the partition $\cP_a$ is contained in collection (\ref{eq:argmin-problem}), and the partition $\cP_b$ is contained in collection (\ref{eq:argmax-problem}). It remains to show that the first and last partitions in the sequence are computable in polynomial-time, which is immediate from Lemma \ref{lem:partitions-at-left-and-right-breakpoints} and polynomial-time tractability of submodular minimization.

\end{proof}

\subsection{Proof of Theorem \ref{thm:st-pps}}\label{sec:relaxing-strict-submodularity}
{We now complete the proof of Theorem \ref{thm:st-pps}. We first state the key technical lemma that allows us to assume that the function is strictly submodular on intersecting pairs.


\begin{lem}\label{lem:unique}
Let $f\colon 2^V\rightarrow \R_{\ge 0}$ be a submodular function and $\varepsilon>0$. 
Let $d_{\varepsilon},e_{\varepsilon}, h^1_{\varepsilon}, h^2_{\varepsilon}:2^V\rightarrow \R_{\ge 0}$ be $d_{\varepsilon}(X):=\epsilon|X||V\setminus X|$, $e_{\varepsilon}(X):=d_{\varepsilon}(X) + \varepsilon\binom{|X|}{2}$, $h^1_{\varepsilon}(X):=f(X)+d_{\varepsilon}(X)$ and $h^2_{\varepsilon}(X):=f(X)+e_{\varepsilon}(X)$ for every $X\subseteq V$. Then, 
\begin{enumerate}[label=(\alph*)]\itemsep0em
    \item\label{it:strict_prop_sym} if $f$ is symmetric, then $h^1_{\varepsilon}$ is symmetric submodular,
    \item\label{it:strict_prop_mon} if $f$ is monotone, then $h^2_{\varepsilon}$ is monotone submodular,
    \item\label{it:strict_prop_poly} given $\varepsilon>0$ and an evaluation oracle for $f$, the function evaluation oracle for $h^1_{\varepsilon}$ and $h^2_{\varepsilon}$ can be implemented in polynomial-time, 
    \item\label{it:strict_prop1} $h^i_{\varepsilon}(A) + h^i_{\varepsilon}(B) > h^i_{\varepsilon}(A\cap B) + h^i_{\varepsilon}(A\cup B)$ for every pair $A, B\subseteq V$ of intersecting sets and $i\in [2]$, and 
    \item\label{it:strict_prop2} if \begin{equation}
        \varepsilon<\frac{1}{4m}\cdot \min \{|f(\cP) - f(\cQ)| \colon \cP,\cQ \text{ are partitions of }V \text{ with }f(\cP)\neq f(\cQ)\}, \label{eq:epsilon}
    \end{equation}
    then, the following holds for each $i\in [2]$:
    if $\mcp$ is an $\{s,t\}$-separating partition such that 
        $h^i_{\varepsilon}(\mcp)\le h^i_{\varepsilon}(\mcq)$ for every $\{s,t\}$-separating partition $\mcq$ with $|\mcq|=|\mcp|$, then $f(\mcp)\le f(\mcq)$ for every $\{s,t\}$-separating partition $\mcq$ with $|\mcq|=|\mcp|$.
\end{enumerate}
\end{lem}

\begin{proof}
    Let $G=(V,E)$ be the complete graph on $V$ and $m\coloneqq\binom{|V|}{2}$. We observe that $d_w$ is the $\varepsilon$-weighted graph cut function of $G$, and $e_w$ is the $\varepsilon$-weighted graph coverage function of $G$. Consequently, both $d_w$ and $e_w$ are posimodular functions. Moreover, both are strictly submodular on intersecting pairs. Furthermore, $d_w$ is symmetric, while $e_w$ is monotone. It follows that if $f$ is symmetric, then $h_\varepsilon^1$ is symmetric, and if $f$ is monotone, then $h_\varepsilon^2$ is monotone. Properties \ref{it:strict_prop_sym}, \ref{it:strict_prop_mon}, and \ref{it:strict_prop1} follow.
    Moreover, both $h_\varepsilon^1$ and $h_\varepsilon^2$ are given as explicit functions of $f,\varepsilon$, meaning a function evaluation oracle can be implemented in polynomial-time, proving property \ref{it:strict_prop_poly}.

    We now show property \ref{it:strict_prop2}. Let $\varepsilon$ satisfy \eqref{eq:epsilon}, $i\in [2]$ and $h=h_{\varepsilon}^i$. 
    Let $\cP$ be an $\{s,t\}$-separating partition such that $h(\cP) \leq h(\cQ)$ for every $\{s,t\}$-separating partition $\cQ$ of size $|\cP|$, and let $\cQ$ be an arbitrary $\{s,t\}$-separating partition with $|\cQ|=|\cP|$. We prove two claims. 

    \begin{claim}\label{cl:tech1}
       If $h(\cP)=h(\cQ)$, then $f(\cP)=f(\cQ)$. 
    \end{claim}
    \begin{proof}
    For the sake of contradiction, suppose that $f(\cP)\neq f(\cQ)$. Then, 
    \begin{align*}
    |f(\cP)-f(\cQ)|
    &=
    |(f(\cP)-f(\cQ)-(h(\cP)-h(\cQ)|&& &&  \text{\hfill (by $h(\cP) = h(\cQ)$)}\\
    &\leq
    |f(\cP)-h(\cP)|+|f(\cQ)-h(\cQ)|\\
    &\leq 
    2m\varepsilon+2m\varepsilon&& &&  \text{\hfill (by $d_\varepsilon(\cP) \leq e_{\varepsilon}(\cP) \leq 2m\varepsilon$ for any $\cP$)}\\
    &<
    |f(\cP)-f(\cQ)|,&& &&  \text{\hfill (by equation \eqref{eq:epsilon})}
    \end{align*}
    a contradiction.
    \end{proof}
    
    \begin{claim}\label{cl:tech2}
        If $h(\cP)<h(\cQ)$, then $f(\cP)\leq f(\cQ)$.
    \end{claim}
    \begin{proof}
    For the sake of contradiction, suppose that $f(\cP)> f(\cQ)$. Then,
    \begin{align*}
        f(\cP) 
        &\leq 
        h(\cP)&& &&  \text{\hfill (by definition)}\\
        &< 
        h(\cQ)&& &&  \text{\hfill (by assumption)}\\
        &\leq
        f(\cQ)+2m\varepsilon&& &&  \text{\hfill (by $d_\varepsilon(\cP) \leq e_{\varepsilon}(\cP) \leq 2m\varepsilon$ for any $\cP$)}\\
        &<
        f(\cQ)+(f(\cP)-f(\cQ))/2&& &&  \text{\hfill (by $f(\cP)> f(\cQ)$)}\\
        &<
        f(\cP),
    \end{align*}
    a contradiction.
    \end{proof}
By Claims~\ref{cl:tech1} and~\ref{cl:tech2}, property~\ref{it:strict_prop2} is satisfied as well.


\end{proof}


Before completing the proof of Theorem \ref{thm:st-pps}, we need one more lemma showing that there exists a polynomial-time algorithm to verify whether a given sequence is an $\{s,t\}$-separating principal partition sequence of a given submodular function $f$.

\begin{lem}\label{lem:verify}
    Given a submodular function $f: 2^V\rightarrow \mathbb{R}$ by its function evaluation oracle along with a sequence of partitions $\cP_1,\cP_2,\ldots, \cP_\ell$ of $V$, there is a polynomial-time algorithm to verify whether $\cP_1,\cP_2,\ldots, \cP_\ell$ is an $\{s,t\}$-separating principal partition sequence of $f$.
\end{lem}
\begin{proof}
    We need only verify the conditions of Definition \ref{def:st-pps}. As argued at the beginning of this section, the breakpoints $\lambda_1,\lambda_2,\dots, \lambda_{\ell-1}$ of any $\{s,t\}$-separating principal partition sequence of $f$ can be found in polynomial-time using the Newton-Dinkelbach method. These will always satisfy property \ref{stpps:1}. Properties \ref{stpps:3} and \ref{stpps:4} are straightforward to verify by checking whether every subsequent partition is a refinement or an $\{s,t\}$-refinement. To show property \ref{stpps:2}, compute the minimum $g^{s,t}(\lambda_i)$ at each breakpoint $i\in [\ell-1]$ using Theorem \ref{thm:min-g}. If there is a $\lambda_i$ such that $g^{s,t}(\lambda_i) < g_{\cP_i}(\lambda_i)$ or $g^{s,t}(\lambda_i) < g_{\cP_{i+1}}(\lambda_i)$, then $\cP_1,\cP_2,\ldots, \cP_\ell$ is not an $\{s,t\}$-separating principal partition sequence of $f$. If there is no such $\lambda_i$, then $g^{s,t}(\lambda_i) = g_{\cP_{i}}(\lambda_i) =  g_{\cP_{i+1}}(\lambda_i) $ for $i\in [\ell-1]$. As the functions $g_{\cP}$ are linear in $\lambda$ and $\cP_i$ is a minimizer at both $\lambda_{i-1}$ and $\lambda_i$ for $2\leq i\leq \ell-2$, it follows that $\cP_i$ is a minimizer for all $\lambda\in [\lambda_{i-1}, \lambda_i]$. Thus property \ref{stpps:2} holds as well, and the sequence is an $\{s,t\}$-separating principal partition sequence of $f$.
\end{proof}

\begin{proof}[Proof of Theorem \ref{thm:st-pps}]
    Given a submodular function $f$ by its evaluation query, we construct the $\{s,t\}$-separating principal partition sequence of $f$ as follows. First, for $\varepsilon_0\coloneqq 1/m$, we use Lemma \ref{lem:unique} to get a function $h_{\epsilon_0}$ that is strictly submodular on intersecting pairs, and then we compute the $\{s,t\}$-separating principal partition sequence of $h_{\epsilon_0}$ using Theorem \ref{lem:strict-st-pps}. If $\varepsilon_0$ does not satisfy \eqref{eq:epsilon}, the computed sequence may not be an $\{s,t\}$-separating principal partition sequence of $f$. We can verify this in polynomial-time using Lemma \ref{lem:verify}. If the computed sequence is not an $\{s,t\}$-separating principal partition sequence of $f$, then we set $\varepsilon_1 \coloneqq \varepsilon_0/2$, and repeat the process.

    After $O(b)$ iterations, where $b$ is the maximum bit complexity of a response to the function evaluation query of $f$, we either compute the $\{s,t\}$-separating principal partition sequence of $f$, or arrive at some $\varepsilon_j$ that satisfies \eqref{eq:epsilon}. 
    Suppose, we arrive at the latter case, i.e., we have a $\varepsilon_j$ that satisfies \eqref{eq:epsilon}. 
    Let $\cP_1,\cP_2,\dots,\cP_\ell$ be the principal partition sequence of $h_{\varepsilon_j}$. 
    Then, $\cP_1,\cP_2,\dots,\cP_\ell$ is also an $\{s,t\}$-separating principal partition sequence of $f$ by the following argument: Properties \ref{stpps:1} and \ref{stpps:4} of Definition \ref{def:st-pps} hold since the sequence is an $\{s,t\}$-separating principal partition sequence of $h$. Properties \ref{stpps:2} and \ref{stpps:3} follow by Lemma \ref{lem:unique}\ref{it:strict_prop2}. 
\end{proof}

\section{Approximation Algorithm for \stkpartition}
\label{sec:monotone}

In this section, we design an algorithm for \stkpartition via $\{s,t\}$-separating principal partition sequences and show that it achieves an approximation factor of $4/3$ for monotone submodular functions and $2$ for posimodular submodular functions. 
We recall \stkpartition below: 
  {\searchprob{\stkpartition}{A submodular function $f\colon 2^V\to \bR_{\ge 0}$ given by a value oracle, distinct elements $s, t \in V$, and $k\in \Z_{\geq 0}$.}{
\[ \min\left\{ \sum_{i=1}^k f(V_i) \colon  \{ V_i\}_{i=1}^k \text{ is a }\{s,t\}\text{-separating partition of } V \text{ into $k$ non-empty parts}\right\}.\]}}

Our algorithm computes an $\{s,t\}$-separating principal partition sequence as mentioned in Theorem~\ref{thm:st-pps}. If the sequence contains a partition with exactly $k$ parts, the algorithm returns this partition. Otherwise, it considers the two partitions $\mcp_{i-1}$ and $\mcp_i$ in the sequence with $|\mcp_{i-1}|<k<|\mcp_i|$. By the properties of the $\{s,t\}$-separating principal partition sequence, we know that $\mcp_i$ is either a refinement of $\mcp_{i-1}$ up to one set or an $\{s,t\}$-refinement of $\mcp_{i-1}$ up to two sets. Suppose first that $\mcp_i$ is a refinement of $\mcp_{i-1}$ up to one set $X\in\mcp_{i-1}$. Then the algorithm proceeds similarly to the $k$-partitioning algorithm of~\cite{NRP96, ppskarthikwang}: it obtains an $\{s,t\}$-separating $k$-partition $\mcp$ from $\mcp_{i-1}$ by replacing $X$ with the $k-|\mcp_{i-1}|$ cheapest parts of $\mcp_i$ contained within $X$ and an additional part containing the remainder. Now suppose that $\mcp_i$ is an $\{s,t\}$-refinement of $\mcp_{i-1}$ up to two sets, namely $X\in \mcp_{i-1}\setminus \mcp_i$ and $Y\in \mcp_i\setminus \mcp_{i-1}$. In this case, the algorithm constructs three $\{s,t\}$-separating $k$-partitions $\sigma_1$, $\sigma_2$, and $\pi$, and returns the one with the smallest objective value, as follows: $\sigma_1$ is obtained from $\mcp_{i-1}$ by replacing $X$ with the $k-|\mcp_{i-1}|$ cheapest parts of $\mcp_i$ contained within $X$ and an additional part containing the remainder; $\sigma_2$ is obtained from $\mcp_{i-1}$ by replacing $X$ with the $k-|\mcp_{i-1}|-1$ cheapest parts of $\mcp_i$ contained within $X$, a part $X\cap Y$, and an additional part containing the remainder; $\pi$ is obtained from $\mcp_i$ by replacing the $|\mcp_i|-k+1$ most expensive parts of $\mcp_i$ contained within $X$ by their union. The algorithm is presented as Algorithm~\ref{algo:stkpartition}. We discuss two examples based on Figures~\ref{fig:refinement} and~\ref{fig:st-refinement} to illustrate the partitions created by the algorithm.

\begin{ex}
First, consider the two partitions in Figure \ref{fig:refinement}. Suppose $\mcp_{i-1}=\{X, C_1, C_2, C_3, C_4, C_5, C_6,\allowbreak D_1,\allowbreak D_2, D_3, D_4\}$ (solid lines), $\mcp_i=\{B_1, B_2, B_3, C_1, C_2, C_3, C_4, C_5, C_6, D_1, D_2, D_3, D_4\}$ (dashed lines), and $k=12$. Suppose $f(B_1)\le f(B_2) \le f(B_3)$. Then, the partition returned by the algorithm is $\mcp=\{C_1, C_2, C_3, C_4, C_5, C_6, D_1, D_2, D_3, D_4, B_1,\allowbreak B_2\cup B_3\}$. 

Next, consider the two partitions in Figure \ref{fig:st-refinement}. Suppose $\mcp_{i-1}=\{X, C_1, C_2, D_1, D_2, D_3, D_4\}$ (solid lines), $\mcp_i=\{Y, B_1, B_2, B_3, B_4, B_5, B_6, D_1, D_2, D_3, D_4\}$ (dashed lines), and $k=9$. Suppose $f(B_1)\le f(B_2) \le \dots \le f(B_6)$. 
Then, the partitions created by the algorithm are $\sigma_1=\{C_1, C_2, D_1, D_2, D_3, D_4, B_1, B_2, (X\cap Y)\cup \bigcup_{j=3}^6 B_j\}$, $\sigma_2=\{C_1, C_2, D_1, D_2, D_3, D_4, B_1,\allowbreak X\cap Y, \bigcup_{j=2}^6 B_j\}$, and $\pi=\{Y, D_1, D_2, D_3, D_4, B_1, B_2, B_3, \bigcup_{j=4}^6 B_j\}$.    
\end{ex}

\begin{algorithm}[!t]
\caption{Approximation Algorithm for $\stkpartition$.}\label{algo:stkpartition}
\begin{algorithmic}[1]
\Statex \textbf{Input: }A submodular function $f\colon 2^V\to\R$ given by value oracle and an integer $k\geq 2$.
\Statex \textbf{Output: }An $\{s,t\}$-separating $k$-partition $\mcp$ of $V$.

\State Use Theorem~\ref{thm:st-pps} to compute an $\{s,t\}$-separating principal partition sequence $\mcp_1, \ldots, \mcp_\ell$ of the submodular function $f$.
\If{$\exists\ j\in [\ell]\colon |\mcp_j|=k$}
    \State Return $\mcp\coloneq \mcp_j$.
\EndIf
\State Let $i\in \{2, \ldots, \ell\}$ such that $|\mcp_{i-1}|<k<|\mcp_i|$.
\If{$\mcp_i$ is a refinement of $\mcp_{i-1}$ up to one set}
    \State Let $X\in \mcp_{i-1}$ be the part refined by $\mcp_i$. 
    \State Let $\mcp'$ be the parts of $\mcp_i$ that are contained in $X$.
    \State Let $\mcp'\coloneq \{B_1, \ldots, B_{|\mcp'|}\}$ such that $f(B_1)\le \ldots \le f(B_{|\mcp'|})$.
    \State \Return $\mcp\coloneq \left(\mcp_{i-1}\setminus \{X\}\right)\cup\bigl\{B_i\colon i\in \left[k-|\mcp_{i-1}|\right]\bigr\}\cup\left\{\bigcup_{j=k-|\mcp_{i-1}|+1}^{|\mcp'|}B_j\right\}$.
\ElsIf{$\mcp_i$ is an $\{s,t\}$-refinement of $\mcp_{i-1}$ up to two sets}
    \State Let $X\in \mcp_i\setminus \mcp_{i-1}$ and $Y\in \mcp_{i-1}\setminus \mcp_i$ such that $\mcp_i$ is an $\{s,t\}$-refinement of $\mcp_{i-1}$ along $(X, Y)$.
    \State Let $\mcp'$ be the parts of $\mcp_i$ that are contained in $X$.
    \State Let $\mcp'\coloneq \{B_1,\ldots, B_{|\mcp'|}\}$ such that $f(B_1)\le \ldots \le f(B_{|\mcp'|})$.
    \State Compute the following three partitions:
    \begin{align*}
        \sigma_1 &\coloneq  \left(\mcp_{i-1}\setminus \{X\}\right)\cup\left\{B_i\colon i\in \left[k-|\mcp_{i-1}|\right]\right\}\cup\bigl\{\left(X\cap Y\right)\cup \bigcup\nolimits_{j=k-|\mcp_{i-1}|+1}^{|\mcp'|}B_j\bigr\}.\\
        \sigma_2 &\coloneq  \left(\mcp_{i-1}\setminus \{X\}\right)\cup\left\{B_i\colon i\in \left[k-|\mcp_{i-1}|-1\right]\right\}\cup\{X\cap Y\}\cup\bigl\{\bigcup\nolimits_{j=k-|\mcp_{i-1}|}^{|\mcp'|}B_j\bigr\}.\\
        \pi &\coloneq  \left(\mcp_i\setminus\left\{B_i\colon |\mcp_i|-k+1\leq i\leq |\mcp'|\right\}\right)\cup\bigl\{\bigcup\nolimits_{j=|\mcp_i|-k+1}^{|\mcp'|}B_j\bigr\}.
    \end{align*}
    \State \Return $\mcp\coloneq \arg\min\{f(\sigma_1), f(\sigma_2), f(\pi)\}$.
\EndIf

\end{algorithmic}
\end{algorithm}

Since an $\{s,t\}$-separating principal partition sequence can be computed in polynomial time by Theorem~\ref{thm:st-pps}, Algorithm~\ref{algo:stkpartition} can indeed be implemented to run in polynomial time. Moreover, the algorithm returns an $\{s,t\}$-separating $k$-partition. The rest of the section is devoted to bounding the approximation factor. 

Our first lemma shows an easy case under which the algorithm identifies the optimum. This lemma has appeared in the literature before in the context of submodular $k$-partitioning (e.g., see~\cite{NRP96}). We include a proof since our problem of interest is the $\{s,t\}$-separating submodular $k$-partitioning problem.

\begin{lem}\label{claim:exact-k}
    Let $\mcp^*$ be an optimal $\{s,t\}$-separating $k$-partition and let $\mcp_1, \ldots, \mcp_\ell$ be an $\{s,t\}$-separating principal partition sequence of $f$. If there exists $j\in [\ell]$ such that $|\mcp_j|=k$, then $f(\mcp_j)\le f(\mcp^*)$.
\end{lem}
\begin{proof}
    Let $\lambda_j$ be a value such that $g^{s,t}(\lambda_j)=g_{\mcp_j}(\lambda_j)$ - such a value exists since $\mcp_j$ is a partition in an $\{s,t\}$-separating principal partition sequence. Then, we have the following:
    \[
    f(\mcp^*)-\lambda_j k =g_{\mcp^*}(\lambda_j) \ge g_{\mcp_j}(\lambda_j) =f(\mcp_j)-\lambda_j|\mcp_j| = f(\mcp_j)-\lambda_j k,
    \]
    yielding $f(\mcp^*)\ge f(\mcp_j)$.
\end{proof}

Let $\mcp^*$ be an optimal $\{s,t\}$-separating $k$-partition. Consider an $\{s,t\}$-separating principal partition sequence $\mcp_1,\ldots,\mcp_\ell$ of the submodular function $f$ with critical values $\lambda_1,\ldots,\lambda_{\ell-1}$, such that there exists no $j\in[\ell]$ with $|\mcp_j|=\ell$. Choose $i\in\{2,\ldots,\ell\}$ satisfying $|\mcp_{i-1}|<k<|\mcp_i|$. Let $\mcp', B_1,\ldots,B_{|\mcp'|}$ be defined as in Algorithm~\ref{algo:stkpartition}, and let $\mcp$ be the partition returned by the algorithm.

We begin with a lower bound for the optimum objective value. This lemma has also appeared in the literature before in the context of submodular $k$-partitioning (e.g., see~\cite{ppskarthikwang}). Our proof in the context of $\{s,t\}$-separating $k$-partitioning problem is similar.

\begin{lem}\label{lem:opt-lower-bounds}
We have the following:
    \begin{enumerate}\itemsep0em
        \item 
        $f(\mcp^*)\ge \left(\frac{|\mcp_i|-k}{|\mcp_i|-|\mcp_{i-1}|}\right)f(\mcp_{i-1}) + \left(\frac{k-|\mcp_{i-1}|}{|\mcp_i|-|\mcp_{i-1}|}\right)f(\mcp_i)$.\label{bound:1}
        \item $f(\mcp^*)\ge f(\mcp_{i-1})$.\label{bound:2}
    \end{enumerate}
\end{lem}
\begin{proof}
For \eqref{bound:1}, let $\lambda_{i-1}$ be the value such that $g_{\mcp_{i-1}}(\lambda_{i-1})=g_{\mcp_{i}}(\lambda_{i-1})=g^{s,t}(\lambda_{i-1})$. Then, we have
\begin{align*}
    f(\mcp_{i-1})-\lambda_{i-1}|\mcp_{i-1}|=f(\mcp_{i})-\lambda_{i-1}|\mcp_{i}|,
\end{align*}
yielding $\lambda_{i-1}=\frac{f(\mcp_i)-f(\mcp_{i-1})}{|\mcp_i|-|\mcp_{i-1}|}$. By the definition of $\{s,t\}$-separating principal partition sequences, we also have
\begin{align*}
    f(\mcp^\ast)-\lambda_{i-1}\cdot k&=g_{\mcp^\ast}(\lambda_{i-1})\geq g^{s,t}(\lambda_{i-1})=g_{\mcp_i}(\lambda_{i-1})= f(\mcp_i)-\lambda_{i-1}|\mcp_i|,
\end{align*}
which implies $f(\mcp^\ast)\geq f(\mcp_i)+\lambda_{i-1}(k-|\mcp_i|)$. Combining these observations, we get
\begin{equation*}
    f(\mcp^\ast)\geq f(\mcp_i)+\frac{f(\mcp_i)-f(\mcp_{i-1})}{|P_i|-|P_{i-1}|}(k-|\mcp_i|)=\frac{|\mcp_i|-k}{|\mcp_i|-|\mcp_{i-1}|}f(\mcp_{i-1})+\frac{k-|\mcp_{i-1}|}{|\mcp_i|-|\mcp_{i-1}|}f(\mcp_i).
\end{equation*}

Now we prove \eqref{bound:2}. Let $P_1^*, \ldots, P_k^*$ denote the parts of $\mcp^{\ast}$ and set $k'\coloneq |\mcp_{i-1}|$. Note that $k'<k$. Consider the $k'$-partition $\mcq$ obtained as $Q_1\coloneq P_1^*,\ldots, Q_{k'-1}\coloneq P_{k'-1}^*, Q_{k'}\coloneq \cup_{j=k'}^{k}P_j^*$. Then, using submodularity and $f(\emptyset)\ge 0$, we get 
\[
f(\mcp^{\ast})
= \sum_{i=1}^{k} f(P_i^*)
\ge \bigl(\sum_{i=1}^{k'-1} f(P_i^*)\bigr)+f\bigl(\bigcup_{j=k'}^{k}P_j^*\bigr)
= \sum_{i=1}^{k'}f(Q_i) 
= f(\mcq). 
\] 
That is, $\mcq$ is a $k'$-partition, while $\mcp_{i-1}$ is an optimal $k'$-partition by Lemma~\ref{claim:exact-k}, meaning that $f(\mcq) \ge f(\mcp_{i-1})$. These observations together imply $f(\mcp^\ast) \ge f(\mcp_{i-1})$. 
\end{proof}

We need the following proposition about posimodular submodular functions from~\cite{ppskarthikwang}. 

\begin{prop}[Chandrasekaran, Wang]\label{prop:sym-posi-ineq}
    Let $f\colon 2^V\to \R_{\ge 0}$ be a submodular function on a ground set $V$. If $f$ is posimodular, then 
    \[
    f(T)\le f(S) + f(S\setminus T)
    \]
    for all $T\subseteq S\subseteq V$.
\end{prop}

Next, we show two upper bounds on the objective value of the solution returned by the algorithm. Variants of this lemma have appeared in the literature before in the context of \kpartition--see~\cite{ppskarthikwang}. Our main contribution is proving similar upper bounds in the context of \stkpartition--recall that the problem has an additional constraint and consequently, the algorithm is also different. We note that the structure of $\{s,t\}$-separating principal partition sequence is more complicated than that of a principal partition sequence. 

\begin{lem}\label{lem:alg-upper-bounds}
    We have the following:
    \begin{enumerate}\itemsep0em
        \item $f(\mcp)\le f(\mcp_i)$.\label{ineq:1}
        \item If $f$ is posimodular, then 
        $f(\mcp)\le f(\mcp_{i-1}) + 2\left(\frac{k-|\mcp_{i-1}|}{|\mcp_i|-|\mcp_{i-1}|+1}\right)f(\mcp_i)$.\label{ineq:2}
        \item If $f$ is monotone, then $f(\mcp)\le f(\mcp_{i-1}) + \left(\frac{k-|\mcp_{i-1}|}{|\mcp_i|-|\mcp_{i-1}|+1}\right)f(\mcp_i)$.\label{ineq:3} 
        \end{enumerate}
\end{lem}
\begin{proof}
We consider two cases. 
\medskip

\noindent \textbf{Case 1.} $\mcp_i$ is a refinement of $\mcp_{i-1}$ up to one set. 
In this case, the argument proceeds analogously to the proof of the algorithm for submodular $k$-partition via principal partition sequences in~\cite{ppskarthikwang}. We include the proof here for completeness. Recall that $\mcp\coloneq (\mcp_{i-1}\backslash\{S\})\cup\left\{B_1,\ldots,B_{k-|\mcp_{i-1}|}\right\}\cup\left\{\bigcup_{j=k-|\mcp_{i-1}|+1}^{|\mcp'|}B_j\right\}$. 

We have
    \begin{align*}
    f(\mcp)
    &=f(\mcp_{i-1})-f(S)+\sum_{j=1}^{k-|\mcp_{i-1}|}f(B_j)+f\bigl(\bigcup\nolimits_{j=k-|\mcp_{i-1}|+1}^{|\mcp'|}B_j\bigr)&&  \\
    &\le f(\mcp_{i-1})-f(S)+\sum_{j=1}^{|\mcp'|}f(B_j) &&  \text{\hfill \hspace{-2cm}(by submodularity)}\\
    &=f(\mcp_i). &&  \text{\hfill \hspace{-2cm}(since $\mcp_i=(\mcp_{i-1}\setminus \{S\})\cup \{B_1, \ldots, B_{|\mcp'|}\}$},
    \end{align*}
showing \eqref{ineq:1}. If $f$ is posimodular, we get
    \begin{align*}
    f(\mcp)
    &=f(\mcp_{i-1})-f(S)+\sum_{j=1}^{k-|\mcp_{i-1}|}f(B_j)+f\bigl(\bigcup\nolimits_{j=k-|\mcp_{i-1}|+1}^{|\mcp'|}B_j\bigr) &&  \\
    &\le f(\mcp_{i-1})+2\sum_{j=1}^{k-|\mcp_{i-1}|}f(B_j) && &&  \text{\hfill \hspace{-2cm}(by Proposition~\ref{prop:sym-posi-ineq} and submodularity)}\\
    &\le f(\mcp_{i-1})+2\bigl(\frac{k-|\mcp_{i-1}|}{|\mcp'|}\bigr)f(\mcp') &&  &&  \text{\hfill \hspace{-2cm}(by the choice of $B_1,\ldots,B_{k-|\mcp_{i-1}|}$)}\\
    &\le f(\mcp_{i-1})+2\bigl(\frac{k-|\mcp_{i-1}|}{|\mcp'|}\bigr)f(\mcp_i) && &&  \text{\hfill \hspace{-2cm}(since $B_1, \ldots, B_{|\mcp'|}$ are parts in $\mcp_i$)}\\
    &= f(\mcp_{i-1})+2\bigl(\frac{k-|\mcp_{i-1}|}{|\mcp_{i}|-|\mcp_{i-1}|+1}\bigr)f(\mcp_i), &&
    \end{align*}
proving \eqref{ineq:2}. Finally, if $f$ is monotone, then 
    \begin{align*}
    f(\mcp)
    &=f(\mcp_{i-1})-f(S)+\sum_{j=1}^{k-|\mcp_{i-1}|}f(B_j)+f\bigl(\bigcup\nolimits_{j=k-|\mcp_{i-1}|+1}^{|\mcp'|}B_j\bigr)&&  \\
    &\le f(\mcp_{i-1})+\sum_{j=1}^{k-|\mcp_{i-1}|}f(B_j) &&  \text{\hfill \hspace{-2cm}(by monotonicity)}\\
    &\le f(\mcp_{i-1})+\frac{k-|\mcp_{i-1}|}{|\mcp'|}f(\mcp') &&  \text{\hfill \hspace{-2cm}(by the choice of $B_1,\ldots,B_{k-|\mcp_{i-1}|}$)}\\
    &\le f(\mcp_{i-1})+\frac{k-|\mcp_{i-1}|}{|\mcp'|}f(\mcp_i) &&  \text{\hfill \hspace{-2cm}(since $B_1, \ldots, B_{|\mcp'|}$ are parts in $\mcp_i$)}\\
    &= f(\mcp_{i-1})+\bigl(\frac{k-|\mcp_{i-1}|}{|\mcp_{i}|-|\mcp_{i-1}|+1}\bigr)f(\mcp_i),
    \end{align*}
yielding \eqref{ineq:3}.
\medskip    

\noindent \textbf{Case 2.} $\mcp_{i}$ is an $\{s,t\}$-refinement of $\mcp_{i-1}$ up to two sets. 
This case is non-trivial and our proof is different from that of the previous one. Let $X\in \mcp_i\setminus \mcp_{i-1}$ and $Y\in \mcp_{i-1}\setminus \mcp_i$ such that $\mcp_i$ is an $\{s,t\}$-refinement of $\mcp_{i-1}$ along $(X, Y)$.  

We have
    \begin{align*}
        f(\mcp) 
        &\le f(\pi)&&  \\
        &= f(\mcp_i) - \sum_{j=|\mcp_i|-k+1}^{|\mcp'|}f(B_j)+f\bigl(\bigcup\nolimits_{j=|\mcp_i|-k+1}^{|\mcp'|}B_j\bigr)&&  \\
        &\le f(\mcp_i), &&  \text{\hfill \hspace{-2cm}(by submodularity)}
    \end{align*}
showing \eqref{ineq:1}. If $f$ is posimodular then, by submodularity and Proposition~\ref{prop:sym-posi-ineq} applied to $T\coloneq (X\cap Y)\cup \bigl(\bigcup_{j=k-|\mcp_{i-1}|+1}^{|\mcp'|}B_j\bigr)$ and $S\coloneq X$, we get 
    \begin{align}
    f\bigl({(X\cap Y)}\cup \bigl(\bigcup\nolimits_{j=k-|\mcp_{i-1}|+1}^{|\mcp'|}B_j\bigr)\bigr) 
    &\le f(X) + f\bigl(\bigcup\nolimits_{j=1}^{k-|\mcp_{i-1}|}B_j\bigr)\nonumber\\
    &\le f(X) + \sum_{j=1}^{k-|\mcp_{i-1}|}f(B_j). \label{ineq:sym-posi-intermed}
    \end{align}
Observe that 
    \begin{align}
        |\mcp'| 
        &= |\{B\in \mcp_i\colon B\subseteq X\}| \notag\\
        &\ge |\{B\in \mcp_i\colon B\subseteq X\}|-|\{A\in \mcp_{i-1}\colon A\subseteq Y\}| +1\notag\\
        &=|\mcp_i|-|\mcp_{i-1}|+1. \label{eq:size-of-P'-posi}
    \end{align}
    Hence, we have that 
    \begin{align*}
    f(\mcp)
        &\le f(\sigma_1) &&  \\
        &=f(\mcp_{i-1})-f(X)+ \sum_{j=1}^{k-|\mcp_{i-1}|} f(B_j) + f\bigl( {(X\cap Y)}\cup \bigl(\bigcup\nolimits_{j=k-|\mcp_{i-1}|+1}^{|\mcp'|}B_j\bigr)\bigr) &&  \\
        &\le f(\mcp_{i-1})+ 2\sum_{j=1}^{k-|\mcp_{i-1}|} f(B_j) &&  \text{\hfill \hspace{-2cm}(by inequality \eqref{ineq:sym-posi-intermed})} \\
        &\le f(\mcp_{i-1})+ 2\bigl(\frac{k-|\mcp_{i-1}|}{|\mcp'|}\bigr)\sum_{j=1}^{|\mcp'|} f(B_j) &&  \text{\hfill \hspace{-2cm}(by definition of $B_1,\ldots, B_{k-|\mcp_{i-1}|}$)} \\
        &= f(\mcp_{i-1})+ 2\bigl(\frac{k-|\mcp_{i-1}|}{|\mcp'|}\bigr)f(\mcp') &&  \\
        &\le f(\mcp_{i-1})+ 2\bigl(\frac{k-|\mcp_{i-1}|}{|\mcp'|}\bigr)f(\mcp_i)&&  \text{\hfill \hspace{-2cm}(since $\mcp'\subseteq \mcp_i$)}\\
        &\le 2f(\mcp_{i-1})+ \bigl(\frac{k-|\mcp_{i-1}|}{|\mcp_i|-|\mcp_{i-1}|+1}\bigr)f(\mcp_i), &&  \text{\hfill \hspace{-2cm}(by inequality \eqref{eq:size-of-P'-posi})} 
    \end{align*}
    proving ~\eqref{ineq:2}. Finally, assume that $f$ is monotone. Let $\sigma_1$ and $\sigma_2$ be as defined in Algorithm~\ref{algo:stkpartition}. Denote by $A_1,\ldots,A_{|\mcp'|+1}$ an ordering of $\mcp'\cup\{X\cap Y\}$ such that $f(A_1)\le \ldots \le f(A_{|\mcp'|+1})$. Let 
    \[
    \sigma \coloneq  \bigl(\mcp_{i-1}\setminus \{X\}\bigr)\cup \bigl\{A_i\colon i\in \left[k-|\mcp_{i-1}|\right]\bigr\}\cup \bigl\{\bigcup\nolimits_{j=k-|\mcp_{i-1}|+1}^{|\mcp'|+1}A_j\bigr\}.
    \]
    Then, $f(\sigma)=\min\{f(\sigma_1), f(\sigma_2)\}$. Observe that 
    \begin{align}
        |\mcp'| 
        &= |\{B\in \mcp_i\colon B\subseteq X\}| \notag\\
        &\ge |\{B\in \mcp_i\colon B\subseteq X\}|-|\{A\in \mcp_{i-1}\colon A\subseteq Y\}| \notag\\
        &=|\mcp_i|-|\mcp_{i-1}|. \label{eq:size-of-P'}
    \end{align}
    Thus we have
    \begin{align*}
        f(\mcp)
        &\le \min\{f(\sigma_1), f(\sigma_2)\}&&  \\
        &=f(\sigma)&&  \\
        &=f(\mcp_{i-1})-f(X)+ \sum_{j=1}^{k-|\mcp_{i-1}|} f(A_j) + f\bigl(\bigcup\nolimits_{j=k-|\mcp_{i-1}|+1}^{|\mcp'|+1}A_j\bigr)&&  \\
        &\le f(\mcp_{i-1})+ \sum_{j=1}^{k-|\mcp_{i-1}|} f(A_j) && \text{\hfill \hspace{-2cm}(by monotonicity and $\bigcup\nolimits_{j=k-|\mcp_{i-1}|+1}^{|\mcp'|+1}A_j\subseteq X$)}  \\
        &\le f(\mcp_{i-1})+ \bigl(\frac{k-|\mcp_{i-1}|}{|\mcp'|+1}\bigr)\sum_{j=1}^{|\mcp'|+1} f(A_j) && \text{\hfill \hspace{-2cm}(by definition of $A_1, \ldots, A_{k-|\mcp_{i-1}|}$)}  \\
        &= f(\mcp_{i-1})+ \bigl(\frac{k-|\mcp_{i-1}|}{|\mcp'|+1}\bigr)\bigl( f(\mcp')+f(X\cap Y) \bigr)&&  \\
        &\le f(\mcp_{i-1})+ \bigl(\frac{k-|\mcp_{i-1}|}{|\mcp'|+1}\bigr)\bigl( f(\mcp')+f(Y) \bigr) &&   \text{\hfill \hspace{-2cm}(by monotonicity)}  \\
        &\le f(\mcp_{i-1})+ \bigl(\frac{k-|\mcp_{i-1}|}{|\mcp'|+1}\bigr)f(\mcp_i) && \text{\hfill \hspace{-2cm}(since $Y\in \mcp_i$ and $\mcp'\subseteq \mcp_i$)}  \\
        &\le f(\mcp_{i-1})+ \bigl(\frac{k-|\mcp_{i-1}|}{|\mcp_i|-|\mcp_{i-1}|+1}\bigr) f(\mcp_i),   && \text{\hfill \hspace{-2cm}(by inequality \eqref{eq:size-of-P'})}
    \end{align*}
    yielding \eqref{ineq:3}.
\end{proof}

We next refine the upper bounds established in Lemma~\ref{lem:alg-upper-bounds} using the lower bounds of Lemma~\ref{lem:opt-lower-bounds}.

    \begin{lem}\label{lem:intermediate-upper-bounds}
    We have the following: 
        \begin{enumerate}\itemsep0em
            \item $f(\mcp)\le \left(\frac{|\mcp_i|-|\mcp_{i-1}|}{k-|\mcp_{i-1}|}\right)f(\mcp^*)-\left(\frac{|\mcp_i|-k}{k-|\mcp_{i-1}|}\right)f(\mcp_{i-1})$. \label{refine:1} 
            \item If $f$ is posimodular, then $f(\mcp)\le \left(\frac{|\mcp_i|-|\mcp_{i-1}|}{|\mcp_i|-|\mcp_{i-1}|+1}\right)
\left(2f(\mcp^*)+ \left(\frac{2k-|\mcp_{i}|-|\mcp_{i-1}|+1}{|\mcp_i|-|\mcp_{i-1}|}\right)f(\mcp_{i-1})\right)$. \label{refine:2} 
            \item If $f$ is monotone, then $f(\mcp)\le \left(\frac{|\mcp_i|-|\mcp_{i-1}|}{|\mcp_i|-|\mcp_{i-1}|+1}\right)\left(f(\mcp^*)+\left(\frac{k-|\mcp_{i-1}|+1}{|\mcp_i|-|\mcp_{i-1}|}\right)f(\mcp_{i-1})\right)$. \label{refine:3}
            \end{enumerate}
    \end{lem}
    \begin{proof}
    We have
    \begin{align*}
    f(\mcp) 
    &\le f(\mcp_i) && \text{\hfill \hspace{-3.5cm}(by Lemma~\ref{lem:alg-upper-bounds}\eqref{ineq:1})}\\
    &= \bigl(\frac{|\mcp_i|-|\mcp_{i-1}|}{k-|\mcp_{i-1}|}\bigr)\bigl(\bigl(\frac{|\mcp_i|-k}{|\mcp_i|-|\mcp_{i-1}|}\bigr)f(\mcp_{i-1})+\bigl(\frac{k-|\mcp_{i-1}|}{|\mcp_i|-|\mcp_{i-1}|}\bigr)f(\mcp_i)\bigr)-\bigl(\frac{|\mcp_i|-k}{k-|\mcp_{i-1}|}\bigr)f(\mcp_{i-1})\\
    &\le \bigl(\frac{|\mcp_i|-|\mcp_{i-1}|}{k-|\mcp_{i-1}|}\bigr)f(\mcp^*)-\bigl(\frac{|\mcp_i|-k}{k-|\mcp_{i-1}|}\bigr)f(\mcp_{i-1}), && \text{\hfill \hspace{-3.5cm}(by Lemma~\ref{lem:opt-lower-bounds}\eqref{bound:1})}
\end{align*}
showing \eqref{refine:1}. If $f$ is posimodular, we get 
\begin{align*}
    f(\mcp)
    &\le  f(\mcp_{i-1})+2\bigl(\frac{k-|\mcp_{i-1}|}{|\mcp_{i}|-|\mcp_{i-1}|+1}\bigr)f(\mcp_i)&& \text{\hfill \hspace{-2cm}(by Lemma~\ref{lem:alg-upper-bounds}\eqref{ineq:2})} \\    
    &= \bigl(\frac{|\mcp_i|-|\mcp_{i-1}|}{|\mcp_i|-|\mcp_{i-1}|+1}\bigr)
    \bigl(\bigl(\frac{|\mcp_i|-|\mcp_{i-1}|+1}{|\mcp_i|-|\mcp_{i-1}|}\bigr)f(\mcp_{i-1})+2\bigl(\frac{k-|\mcp_{i-1}|}{|\mcp_i|-|\mcp_{i-1}|}\bigr)f(\mcp_i)\bigr)&& \\
    &\le \bigl(\frac{|\mcp_i|-|\mcp_{i-1}|}{|\mcp_i|-|\mcp_{i-1}|+1}\bigr)
    \bigl(2f(\mcp^*)+ \bigl(\frac{2k-|\mcp_{i}|-|\mcp_{i-1}|+1}{|\mcp_i|-|\mcp_{i-1}|}\bigr)f(\mcp_{i-1})\bigr), && \text{\hfill \hspace{-2cm}(by Lemma~\ref{lem:opt-lower-bounds}\eqref{bound:1})}
\end{align*}
proving \eqref{refine:2}. Finally, if $f$ is monotone, then
\begin{align*}
    f(\mcp)
    &\le  f(\mcp_{i-1})+\bigl(\frac{k-|\mcp_{i-1}|}{|\mcp_{i}|-|\mcp_{i-1}|+1}\bigr)f(\mcp_i)&& \text{\hfill \hspace{-2cm}(by Lemma~\ref{lem:alg-upper-bounds}\eqref{ineq:3})}\\
    &= \bigl(\frac{|\mcp_i|-|\mcp_{i-1}|}{|\mcp_i|-|\mcp_{i-1}|+1}\bigr)\bigl(\bigl(\frac{|\mcp_i|-|\mcp_{i-1}|+1}{|\mcp_i|-|\mcp_{i-1}|}\bigr)f(\mcp_{i-1})+\bigl(\frac{k-|\mcp_{i-1}|}{|\mcp_i|-|\mcp_{i-1}|}\bigr)f(\mcp_i)\bigr)\\
    & \le  \bigl(\frac{|\mcp_i|-|\mcp_{i-1}|}{|\mcp_i|-|\mcp_{i-1}|+1}\bigr)\bigl(f(\mcp^*)+\bigl(\frac{k-|\mcp_{i-1}|+1}{|\mcp_i|-|\mcp_{i-1}|}\bigr)f(\mcp_{i-1})\bigr),&& \text{\hfill \hspace{-2cm}(by Lemma~\ref{lem:opt-lower-bounds}\eqref{bound:1})}
    \end{align*}
    yielding \eqref{refine:3}.
\end{proof}

We now bound the approximation factor of Algorithm~\ref{algo:stkpartition}. Using the bounds in Lemma~\ref{lem:intermediate-upper-bounds}, the rest of the argument is identical to the argument in~\cite{ppskarthikwang} and we repeat it here for the sake of completeness.
We need the following proposition that is implicit in~\cite{ppskarthikwang}.

\begin{prop}[Chandrasekaran, Wang]\label{prop:mono-posi-calculations}
    Suppose $A,B\geq 1$ are integers. Then, for every $c\in\bR_{\geq 0}$, the following two inequalities hold:  
    \begin{align}
    \min \bigl\{\frac{A+B}{A}-\frac{B}{A}\cdot c, \;\bigl(\frac{A+B}{A+B+1}\bigr)\bigl(2+\frac{A-B+1}{A+B}\cdot c\bigr) \bigr\} &\le 2\bigl(1-\frac{1}{A+B+2}\bigr), \text{ and}\label{eqn:min-of-two-terms-sym-posi}\\
    \min \bigl\{\frac{A+B}{A}-\frac{B}{A}\cdot c, \;\bigl(\frac{A+B}{A+B+1}\bigr)\bigl(1+\frac{1+A}{A+B}\cdot c\bigr) \bigr\}
    &\le \frac{4}{3}\bigl(1-\frac{1}{3(A+B)+4}\bigr).
    \label{eqn:min-of-two-terms-monotone}
    \end{align}
\end{prop}

We are now ready to bound the approximation factor of Algorithm~\ref{algo:stkpartition} for posimodular submodular functions and monotone submodular functions.

\begin{thm}\label{thm:mono-approx-factor}
Let $f$ be a submodular function, $k\geq 2$ be an integer, and let $\cP^*$ denote an optimal $\{s,t\}$-separating $k$-partition. Let $\cP$ denote the partition returned by Algorithm~\ref{algo:stkpartition}. 
\begin{enumerate}\itemsep0em 
    \item If $f$ is posimodular, then $f(\mcp)\le 2f(\mcp^*)$. \label{eq:posi}
    \item If $f$ is monotone, then $f(\mcp)\le \frac{4}{3}f(\mcp^*)$. \label{eq:mon}
\end{enumerate}
\end{thm}
\begin{proof}
First, if $f(\mcp^*)=0$, then by Lemma~\ref{lem:opt-lower-bounds}\eqref{bound:2} we also have $f(\mcp_{i-1})\le f(\mcp^*)=0$, and hence $f(\mcp)=0$ by Lemma~\ref{lem:intermediate-upper-bounds}\eqref{refine:1}. Consequently, the returned partition is optimal. We may therefore assume that $f(\mcp^*)>0$. Define $c\coloneq f(\mcp_{i-1})/f(\mcp^*)$, $A\coloneq k-|\mcp_{i-1}|$, and $B\coloneq |\mcp_i|-k$. Note that $A,B\ge1$.

Assume first that $f$ is posimodular. Then, the upper bounds of Lemma~\ref{lem:intermediate-upper-bounds} can be rewritten as 
\begin{align*}
        \frac{f(\mcp)}{f(\mcp^*)}
        &\le \min \bigl\{\frac{A+B}{A}-\frac{B}{A}\cdot c, \;\bigl(\frac{A+B}{A+B+1}\bigr)\bigl(2+\frac{A-B+1}{A+B}\cdot c\bigr) \bigr\}&& \\
    &\le 2\bigl(1-\frac{1}{A+B+2}\bigr) && \text{\hfill \hspace{-2cm}(by Proposition~\ref{prop:mono-posi-calculations})}\\
    &\le 2\bigl(1-\frac{1}{n}\bigr), && \text{\hfill \hspace{-2cm}(since $A+B=|\mcp_i|-|\mcp_{i-1}|\leq n-2$)}
\end{align*}
proving \eqref{eq:posi}.

Now suppose that $f$ is monotone. Then, the upper bounds of Lemma~\ref{lem:intermediate-upper-bounds} can be rewritten as 
\begin{align*}
        \frac{f(\mcp)}{f(\mcp^*)}
        &\le \min \bigl\{\frac{A+B}{A}-\frac{B}{A}\cdot c, \;\bigl(\frac{A+B}{A+B+1}\bigr)\bigl(1+\frac{1+A}{A+B}\cdot c\bigr) \bigr\} &&\\
    &\le\frac{4}{3}\bigl(1-\frac{1}{3(A+B)+4}\bigr)&& \text{\hfill \hspace{-2cm}(by Proposition~\ref{prop:mono-posi-calculations})}\\
    &\leq \frac{4}{3}\bigl(1-\frac{1}{3n-2}\bigr),&& \text{\hfill \hspace{-2cm}(since $A+B=|\mcp_i|-|\mcp_{i-1}|\leq n-2$)}
\end{align*}
proving \eqref{eq:mon}.
\end{proof}

\section{Hypergraph Orientations}
\label{sec:hypergraph}

In this section, we present an application of the polynomial-time computability of $\{s,t\}$-separating principal partition sequence to hypergraph orientations. We recall that a \emph{directed hypergraph} $\dG=(V, E, \head\colon E\rightarrow V)$ is specified by a vertex set $V$, hyperedge set $E$ where each $e\in E$ is a subset of $V$, and a function $\head\colon E\rightarrow V$ with the property that $\head(e)\in e$ for each $e\in E$. For a subset $U\subseteq V$, we define $\delta^{in}_{\dG}(U)\coloneq\{e\in E\colon \head(e)\in U, e\setminus U\neq \emptyset\}$ and the function $d^{in}_{\dG}\colon 2^V\rightarrow\bZ$  defined by $d^{in}_{\dG}(U)\coloneq|\delta^{in}_{\dG}(U)|$. It is well-known that the function $d^{in}_{\dG}$ is submodular. By Menger's theorem, $\dG$ is $k$-hyperarc-connected if and only if $d^{in}_{\dG}(U)\ge k\ \forall\ \emptyset\neq U\subsetneq V$. Moreover, $\dG$ is \textit{$(k, (s,t),\ell)$-hyperarc-connected} if and only if the following two conditions hold:
\begin{align*}
    d^{in}_{\dG}(U)&\ge k\ \forall\ \emptyset\neq U\subsetneq V,\\
    d^{in}_{\dG}(U)&\ge \ell\ \forall\ t\in U\subseteq V-s.
\end{align*}

Frank, Kir\'{a}ly, and Kir\'{a}ly gave a complete characterization for the existence of an orientation $\dG$ of a given hypergraph $G$ with specified vertices $s$ and $t$ such that (1) $\dG$ is $k$-hyperarc-connected, (2) $\dG$ has $k_1$ hyperedge-disjoint paths from $s$ to $t$ and (3) $\dG$ has $k_2$ hyperedge-disjoint paths from $t$ to $s$, where $k_1, k_2\ge k$. Lemma \ref{lem:k1-k2-orientation} in the appendix shows that this problem is equivalent to \kstellConnOrient. 
Frank, Kir\'{a}ly, and Kir\'{a}ly~\cite[Theorem 5.1]{tamas-frank-zoli} showed the following characterization for the existence of a $(k, (s,t),\ell)$-hyperarc-connected orientation. 

\begin{thm}[Frank, Király, Király]\label{thm:st-orientation}
    Let $G = (V, E)$ be a hypergraph, $s, t \in V$, and $k, \ell\in \Z_{\ge 0}$. The hypergraph $G$ has a $(k,(s,t),\ell)$-hyperarc-connected orientation if and only if
\begin{equation} 
|\delta_G(\cP)| \geq \sum_{X\in \cP} p^{s,t}_{k, \ell}(X) \quad \text{for all partition }\cP \text{ of } V, \label{eq:st-partition-connectivity}
\end{equation}
where $\delta_G(\cP)$ is the set of hyperedges in $G$ that intersect at least two parts of $\mcp$ and $p^{s,t}_{k, \ell}\colon 2^V\rightarrow \R$ is defined as 
\begin{equation}\label{eq:conditions}
    p^{s,t}_{k, \ell}(X) = \begin{cases}
    0 \hfill &\text{ if }X=V\text{ or }X=\emptyset, \\
    \max\{k,\ell\}  &\text{ if }t\in X\subseteq V-s,\\
    k \hfill &\text{ otherwise}.
\end{cases} 
\end{equation}
\end{thm}

The proof of Theorem~\ref{thm:st-orientation} given by Frank, Kir\'{a}ly, and Kir\'{a}ly~\cite{tamas-frank-zoli} is not constructive. That is, the proof does not lead to a polynomial-time algorithm to verify whether a given hypergraph $G$ with specified vertices $s,t$ satisfies \eqref{eq:st-partition-connectivity} and if so, then find a $(k, (s,t),\ell)$-hyperarc-connected orientation of $G$. 
We resolve this issue by giving a polynomial-time algorithm in Section \ref{sec:orientation-algo}. In Section~\ref{sec:orientation-min-max}, we address the optimization problems of maximizing $k$ for a given $\ell$ (or maximizing $\ell$ for a given $k$) so that $G$ has a $(k, (s,t),\ell)$-hyperarc-connected orientation, give a min-max relation, and show that the corresponding minimization problem is solvable in polynomial time via our results on $\{s,t\}$-separating principal partition sequence. 
We use the following result due to Frank, Kir\'aly, and Kir\'aly~\cite{tamas-frank-zoli}. We note that this result, in particular, gives a polynomial-time algorithm for \kstellConnOrient for the case $\ell = k$.

\begin{thm}[Frank, Király, Király]\label{thm:FKK-partition-connectivity}
    Let $G=(V, E)$ be a hypergraph and $k\in \Z_{\ge 0}$. Then, $G$ has a $k$-hyperarc-connected orientation if and only if
    \begin{align}
    |\delta_G(\mcp)|&\ge k|\mcp|\ \text{for all partition } \mcp \text{ of }V, \label{ineq:partition-connectivity}
    \end{align}
    where $\delta_G(\mcp)$ is the set of hyperedges in $G$ that intersect at least two parts of $\mcp$. Moreover, there exists a polynomial-time algorithm to verify whether a given hypergraph satisfies \eqref{ineq:partition-connectivity}, and if so, then find a $k$-hyperarc-connected orientation of the given hypergraph. 
\end{thm}

\subsection{Algorithm for \texorpdfstring{$(k,(s,t),\ell)$}{(k,(s,t),l)}-Hyperarc-Connected Orientation Problems}
\label{sec:orientation-algo}

In this section, we give a polynomial-time algorithm to verify whether a given hypergraph $G$ with specified vertices $s,t$ satisfies \eqref{eq:st-partition-connectivity} and if so, then find a $(k, (s,t),\ell)$-hyperarc-connected orientation of $G$. 

\begin{thm}\label{thm:algorithmic-orientation}
    Let $G=(V, E)$ be a hypergraph with specified vertices $s,t\in V$ and $k,\ell\in \Z_{\ge 0}$. 
    \begin{enumerate}\itemsep0em
        \item There exists a polynomial-time algorithm to verify if \eqref{eq:st-partition-connectivity} holds for every partition $\cP$ of $V$ and if not, then return a partition $\cP$ for which \eqref{eq:st-partition-connectivity} is violated. 
        \item If \eqref{eq:st-partition-connectivity} holds for every partition $\mcp$ of $V$, then there exists a polynomial-time algorithm to find a $(k,(s,t),\ell)$-hyperarc-connected orientation. 
    \end{enumerate}
\end{thm}



We first address the problem of verifying whether a given hypergraph $G$ with specified vertices $s,t$ satisfies \eqref{eq:st-partition-connectivity}. 
The following lemma proves the first part of Theorem~\ref{thm:algorithmic-orientation}. 

\begin{lem}\label{lem:st-partition-connectivity-poly-time}
    Given a hypergraph $G=(V, E)$  with specified vertices $s,t\in V$ and $k,\ell\in \Z_{\ge 0}$, and $p^{s,t}_{k, \ell}\colon 2^V\to \bR$ defined as in \eqref{eq:conditions}, there exists a polynomial-time algorithm to solve the following optimization problem: 
    \[
    \min\left\{|\delta_G(\mcp)|-p^{s,t}_{k, \ell}(\mcp)\colon \mcp \text{ is a partition of }V\right\}. 
    \]
\end{lem}
\begin{proof}
    Let $\dG$ be an arbitrary orientation of $G$. We observe that $|\delta_G(\mcp)|=\sum_{P\in \mcp}d^{in}_{\dG}(P)$ for every partition $\mcp$ of $V$.
    Therefore, the optimization problem is equivalent to the following: 
    \begin{align*}
        \min&\left\{d^{in}_{\dG}(\mcp)-p^{s,t}_{k, \ell}(\mcp)\colon \mcp \text{ is a partition of }V\right\}. 
    \end{align*}
    Based on the definition of the function $p^{s,t}_{k, \ell}$, the problem reduces to solving the following two problems. 
    \begin{align*}
        \min&\left\{d^{in}_{\dG}(\mcp)-k|\mcp|\colon \mcp \text{ is a partition of }V\right\}, \text{ and} \label{orientation-problem:partition}\\
        \min&\left\{d^{in}_{\dG}(\mcp)-k(|\mcp|-1)-\ell\colon \mcp \text{ is an }\{s,t\}\text{-separating partition of }V\right\}.
    \end{align*}
    Since, $d^{in}_{\dG}$ is submodular and we have a polynomial-time evaluation oracle for it, the latter two problems 
    can be solved in polynomial time by Theorem~\ref{thm:min-g}. 
\end{proof}

We now prove the second part of Theorem~\ref{thm:algorithmic-orientation}. We need the following lemma that is implicit in~\cite{tamas-frank-zoli}.

\begin{prop}[Frank, Király, Király]\label{prop:vector_covering}
   Let $G=(V,E)$ be a hypergraph and $x\in \Z^V$ a vector such that $x(V) = |E|$ and $x(Y)\ge i_G(Y)$ for every $Y\subseteq V$, where $i_G(Y)$ is the number of hyperedges in $G$ that are contained in $Y$. Then, there exists a polynomial-time algorithm to find an orientation $\dG$ of $G$ with indegree vector $x$. 
\end{prop}

The following lemma concludes the proof of Theorem~\ref{thm:algorithmic-orientation}.

\begin{lem}\label{lem:st-orientation-poly-time}
    Let $G=(V, E)$ be a hypergraph with specified vertices $s,t\in V$ and $k,\ell\in \Z_{\ge 0}$ such that \eqref{eq:st-partition-connectivity} holds for every partition $\mcp$ of $V$. Then, there exists a polynomial-time algorithm to find a $(k, (s,t), \ell)$-hyperarc-connected orientation of $G$. 
\end{lem}
\begin{proof}
We state the algorithm that is implicit in the proof of Frank, Kir\'{a}ly, and Kir\'{a}ly~\cite{tamas-frank-zoli} in Algorithm~\ref{alg:ori}. The correctness of the algorithm follows from the proof in~\cite{tamas-frank-zoli}. We will show here that the algorithm can be implemented to run in polynomial time.

\begin{algorithm}[!t]
\caption{An algorithm for finding a $(k, (s,t),\ell)$-hyperarc-connected orientation}\label{alg:ori}
\begin{algorithmic}[1]
    \Statex \textbf{Input:} A connected hypergraph $G=(V,E)$, two vertices $s, t\in V$, and $k,\ell\in \Z_{\ge 0}$ such that \eqref{eq:st-partition-connectivity} holds.
    \Statex \textbf{Output:} A $(k, (s,t),\ell)$-hyperarc-connected orientation of $G$.
    \For{each $v\in V$}
    \State\label{line:min1} Compute $x_v \coloneq\min\left\{|\delta_{G}(\cP)|-p^{s,t}_{k, \ell}(\cP)\colon \cP \text{ is a partition of }V \text{ with }\{v\}\text{ as a singleton part of }\cP\right\}.$
    \EndFor
    \State\label{line:last}\Return an orientation $\dG$ of $G$ with in-degree vector $x$.
\end{algorithmic}
\end{algorithm}

We now show that Step~\ref{line:min1} can be implemented in polynomial-time. Let $v\in V$. We observe that the problem of computing 
\[
x_v\coloneq\min\left\{|\delta_{G}(\cP)|-p^{s,t}_{k, \ell}(\cP)\colon \cP \text{ is a partition of }V \text{ with }\{v\}\text{ as a singleton part of }\mcp\right\}
\]
reduces to the problem of computing $y_v$ defined as follows, where the hypergraph $G-v$ is obtained from $G$ by deleting the vertex $v$ and all hyperedges incident to $v$:
\[
y_v\coloneq\min\left\{|\delta_{G-v}(\mcp)|-p^{s,t}_{k, \ell}(\mcp)\colon \mcp \text{ is a partition of }V\setminus \{v\}\right\}.
\]
Indeed, this is because $x_v=y_v+|\delta_G(v)|-p^{s,t}_{k, \ell}(\{v\})$. 
Observe that the problem of computing $y_v$ is solvable in polynomial time by Lemma~\ref{lem:st-partition-connectivity-poly-time}.

Next, we show that Step~\ref{line:last} can be implemented to run in polynomial time. Frank, Kir\'{a}ly, and Kir\'{a}ly~\cite{tamas-frank-zoli} showed that, assuming~\eqref{eq:st-partition-connectivity}, the vector $x\in\Z^V$ computed in the first three steps satisfies $x\ge 0$, $x(V)=|E|$, and $x(Y)\ge i_G(Y) + p^{s,t}_{k,\ell}(Y)$ for every $Y \subseteq V$, where $i_G(Y)$ denotes the number of hyperedges of $G$ contained in $Y$. Hence, by Proposition~\ref{prop:vector_covering}, there exists a polynomial-time algorithm that finds an orientation $\dG$ with in-degree vector $x$. Moreover, for every $Y \subseteq V$, we have $d^{in}_{\dG}(Y)=x(Y)-i_G(Y)\ge p^{s,t}_{k,\ell}(Y)$. 
\end{proof}

\subsection{Min-Max Relations for \texorpdfstring{$(k,(s,t),\ell)$}{(k,(s,t),l)}-Hyperarc-Connected Orientation Problems}
\label{sec:orientation-min-max}
In this section, we consider two optimization variants of $(k,(s,t),\ell)$-hyperarc-connected orientation problems, present min-max relations for both, and outline a polynomial-time algorithm to solve the corresponding minimization problem using our results on $\{s,t\}$-separating principal partition sequence.  
Our min-max relations follow immediately from the known characterization for the existence of $(k,(s,t),\ell)$-hyperarc-connected orientation \cite{tamas-frank-zoli}. Our main contribution is a polynomial-time algorithm to solve the minimization problem that arises in the min-max relation. 

Firstly, we consider the problem of maximizing $\ell$ for a given $k$ such that a given hypergraph $G$ with vertices $s,t\in V$ has a $(k, (s,t),\ell)$-connected orientation and give a min-max relation. 
For this maximization problem, we work under the assumption that $G$ has a $k$-hyperarc-connected orientation -- otherwise, $G$ has no $(k,(s,t),\ell)$-hyperarc-connected orientation for every $\ell$. 
\begin{thm}\label{thm:orientation-min-max-ell}
    Let $G=(V, E)$ be a hypergraph, $s,t\in V$, and $k\in \Z_{\ge 0}$ such that $G$ has a $k$-hyperarc-connected-orientation. 
    Then, 
    \begin{align*}
    &\max\left\{\ell\colon G\text{ has a }(k, (s,t),\ell)\text{-hyperarc-connected orientation}\right\}\\
    & \quad =\min\left\{|\delta_G(\mcp)|-k(|\mcp|-1)\colon \mcp\text{ is an }\{s,t\}\text{-separating partition of }V\right\}. 
    \end{align*}
\end{thm}
\begin{proof}
    We know that both the $\max$ and the $\min$ are at least $k$ by Theorem~\ref{thm:FKK-partition-connectivity}. For $\ell\ge k$, by Theorem~\ref{thm:st-orientation}, $G$ has a $(k,(s,t),\ell)$-hyperarc-connected orientation if and only if $|\delta_G(\mcp)|\ge k(|\mcp|-1)+\ell$ for every $\{s,t\}$-separating partition $\mcp$ of $V$. Hence, the theorem follows. 
\end{proof}

Next, we consider the problem of maximizing $k$ for a given $\ell$ such that $G$ has a $(k, (s,t),\ell)$-connected orientation and give a min-max relation. 
For this maximization problem, we work under the assumption that $G$ has $\ell$ hyperedge-disjoint paths between $s$ and $t$ -- otherwise, $G$ has no $(k,(s,t),\ell)$-hyperarc-connected orientation for every $k$.

\begin{thm}\label{thm:orientation-min-max-k}
    Let $G=(V, E)$ be a hypergraph, $s,t\in V$, and $\ell\in \Z_{\ge 0}$ such that $G$ has $\ell$ hyperedge-disjoint paths between $s$ and $t$. 
    Then, 
    \begin{align*}
        \max\left\{k\colon G\text{ has a }(k, (s,t),\ell)\text{-hyperarc-connected orientation}\right\} = \min\{ \alpha, \beta\},
    \end{align*}
    where
\begin{align*}
\alpha & \coloneqq \min\bigl\{\bigl\lfloor\frac{|\delta_G(\mcp)|}{|\mcp|}\bigr\rfloor\colon \mcp \text{ is a partition of }V\bigr\},\ \text{and}\\
\beta & \coloneqq \min\bigl\{\bigl\lfloor\frac{|\delta_G(\mcp)|-\ell}{|\mcp|-1}\bigr\rfloor\colon \mcp \text{ is an }\{s,t\}\text{-separating partition of }V\bigr\}.
\end{align*}
\end{thm}
\begin{proof}
    We first show that $\max \ge \min$. Suppose that $G$ has a $(k,(s,t),\ell)$-hyperarc-connected orientation. Let $\mcp$ be an arbitrary partition of $V$. By Theorem~\ref{thm:FKK-partition-connectivity}, we have $|\delta_G(\mcp)| \ge k|\mcp|$, and hence $\lfloor |\delta_G(\mcp)| / |\mcp| \rfloor \ge k$. Next, suppose that $\mcp$ is an $\{s,t\}$-separating partition of $V$. By Theorem~\ref{thm:st-orientation}, $|\delta_G(\mcp)| \ge k(|\mcp|-1) + \max\{k,\ell\} \ge k(|\mcp|-1) + \ell$, and therefore $\lfloor (|\delta_G(\mcp)| - \ell)/(|\mcp|-1) \rfloor \ge k$.

    Next, we show that $\max\le \min$. Suppose that $\min\ge k$. Then, we have
    \begin{align*}
        |\delta_G(\mcp)|&\ge k(|\mcp|-1) + k\quad \text{for all partition } \mcp \text{ of }V, \text{ and}\\
        |\delta_G(\mcp)|&\ge k(|\mcp|-1) + \ell\quad \text{for all $\{s,t\}$-separating partition } \mcp \text{ of }V. 
    \end{align*}
    Consequently, 
    \begin{align*}
        |\delta_G(\mcp)|&\ge k(|\mcp|-1) + k\quad \text{for all partition } \mcp \text{ of }V, \text{ and}\\
        |\delta_G(\mcp)|&\ge k(|\mcp|-1) + \max\{k,\ell\}\quad \text{for all $\{s,t\}$-separating partition } \mcp \text{ of }V. 
    \end{align*}
    Equivalently, $|\delta_G(\mcp)|\ge p^{s,t}_{k, \ell}(\mcp)$ for every partition $\mcp$ of $V$. Hence, by Theorem~\ref{thm:st-orientation}, the hypergraph $G$ has a $(k, (s,t), \ell)$-hyperarc-connected orientation. 
\end{proof}

We now outline algorithms to solve the minimization problems in Theorems~\ref{thm:orientation-min-max-ell} and~\ref{thm:orientation-min-max-k} in polynomial time. We need some background. 
We recall that a \emph{directed hypergraph} $\dG=(V, E, \head\colon E\rightarrow V)$ is specified by a vertex set $V$, hyperedge set $E$ where each $e\in E$ is a subset of $V$, and a function $\head\colon E\rightarrow V$ with the property that $\head(e)\in e$ for each $e\in E$. For a subset $U\subseteq V$, we define $\delta^{in}_{\dG}(U)\coloneq\{e\in E\colon \head(e)\in U, e\setminus U\neq \emptyset\}$ and the function $d^{in}_{\dG}\colon 2^V\rightarrow\bZ$  defined by $d^{in}_{\dG}(U)\coloneq|\delta^{in}_{\dG}(U)|$. 
It is well-known that the function $d^{in}_{\dG}$ is submodular. 

Let $G=(V, E)$ be a hypergraph and $s,t\in V$. Let $\dG=(V, E, head: E\rightarrow V)$ be an arbitrary orientation of $G$. We observe that for each partition $\mcp$ of $V$, we have that $|\delta_G(\mcp)|=\sum_{U\in \mcp}d^{in}_{\dG}(A)=d^{in}_{\dG}(\mcp)$. We define $g_{\mcp}(\lambda)\coloneq d^{in}_{\dG}(\mcp)-\lambda|\mcp|$ for each partition $\mcp$ of $V$, $g(\lambda)\coloneq \min\{g_{\mcp}(\lambda)\colon \mcp \text{ is a partition of }V\}$, and \\ $g^{s,t}(\lambda)\coloneq\min\{g_{\mcp}(\lambda)\colon \mcp \text{ is a }\{s,t\}\text{-separating partition of }V\}$.

Now, for a given $k\in \Z_{\ge 0}$, the minimization problem in Theorem \ref{thm:orientation-min-max-ell} is equivalent to minimizing $g^{s,t}(k)+k$ and we note that $g^{s,t}(k)$ can be computed in polynomial time using Theorem \ref{thm:st-pps}. 
Next, for a given $\ell\in \Z_{\ge 0}$, we address the minimization problem in Theorem \ref{thm:orientation-min-max-k}. Here, we note that $\alpha$ is the floor of the smallest number $\lambda$ such that $g(\lambda)\le 0$ and $\beta$ is the floor of the smallest number $\lambda$ such that $g^{s,t}(\lambda)\le \ell-\lambda$. 
We can compute $g(\lambda)$ and $g^{s,t}(\lambda)$ for all $\lambda$ in polynomial time using Theorems \ref{thm:PPS-exists-algo} and \ref{thm:st-pps} respectively. Consequently, we can compute $\alpha$ and $\beta$ in polynomial time.

\begin{rem}
If $G=(V,E)$ is a hyperedge-weighted hypergraph with weights $w\colon E \rightarrow \R_{+}$ and $|\delta_G(\mcp)|$ is replaced by $\sum_{e \in \delta_G(\mcp)} w_e$, then the corresponding weighted minimization problems can also be solved in strongly polynomial time using the same approach. For hyperedge-weighted hypergraphs, the maximization orientation problems on the left hand sides can be defined naturally by interpreting the weight of a hyperedge as the number of its parallel copies; we omit the formal definition in the interests of brevity.
\end{rem}

\section{Conclusion}
\label{sec:conc}
Motivated by the breadth of applications of the principal partition sequence of submodular functions, we investigated the notion of $\{s,t\}$-separating principal partition sequence. We illustrated two applications of this sequence: we designed approximation algorithms for \stkpartition for monotone and posimodular submodular functions and polynomial-time algorithms for \kstellConnOrient in hypergraphs. 
A natural direction for research is whether 
variants of principal partitioning sequence provide insights into the approximability of multiway cut. As we have seen, it does provide the current-best approximation factor for \kpartition and \stkpartition for monotone and posimodular submodular functions. 
Another interesting direction is whether the principal partition sequence of the cut function of a given graph in near-linear time.  
It would also be interesting to understand other applications of the $\{s,t\}$-separating principal partition sequence. 
E.g., the principal partition sequence of the graph cut function is related to recursive ideal tree packing \cite{CQX20}; do we have a similar connection for $\{s,t\}$-separating principal partition sequence?

\medskip

\paragraph{Acknowledgement.} Karthekeyan Chandrasekaran was partially supported by NSF grant CCF-2402667. Krist\'{o}f B\'{e}rczi and Tam\'{a}s Kir\'{a}ly were supported in part by the Lend\"ulet Programme of the Hungarian Academy of Sciences -- grant number LP2021-1/2021, by the Ministry of Innovation and Technology of Hungary from the National Research, Development and Innovation Fund -- grant numbers ADVANCED 153096, ADVANCED 150556, and ELTE TKP 2021-NKTA-62, by Dynasnet European Research Council Synergy project -- grant number ERC-2018-SYG 810115

\bibliographystyle{abbrv}
\bibliography{st-sep}

\appendix

\section{Auxiliary Lemmas}

\begin{lem}\label{lem:piecewise-linear}
    The function $g^{s,t}$ is piecewise linear with at most $|V|-2$ breakpoints.
\end{lem}
\begin{proof}
    Each $\{s,t\}$-separating partition $\cP$ must have at least 2 parts, and no more than $|V|$. Thus there are $|V|-1$ possible values of $|\cP|$, and $g^{s,t}$ is the minimum of linear functions with $|V|-1$ different slopes, which is piecewise linear with at most $|V|-1$ pieces. Therefore, $g^{s,t}$ has at most $|V|-2$ breakpoints.
\end{proof}

\begin{restatable}{lem}{lemmak1k2orientations}\label{lem:k1-k2-orientation}
    Let $G=(V,E)$ be a hypergraph, $s,t\in V$, and $k, \ell\in \Z_{\ge 0}$ with $\ell\ge k$. Then, $G$ has a $(k,(s,t),\ell)$-hyperarc-connected orientation if and only if for every $k_1, k_2\in \Z_{\ge k}$ with $k_1+k_2=\ell+k$, there exists an orientation $\dGp$ of $G$ such that 
    \begin{enumerate}\itemsep0em
        \item $\dGp$ is $k$-hyperarc-connected, \label{prop:kac}
        \item $\dGp$ has $k_1$ hyperedge-disjoint paths from $s$ to $t$, or equivalently, $d^{in}_{\dG}(U)\ge k_1$ for all $t\in U\subseteq V-s$, and \label{prop:k1}
        \item $\dGp$ has $k_2$ hyperedge-disjoint paths from $t$ to $s$, or equivalently, $d^{in}_{\dG}(U)\ge k_2$ for all $s\in U\subseteq V-t$. \label{prop:k2}
    \end{enumerate}
\end{restatable}
\begin{proof}
    The backwards direction follows by taking $k_1 \coloneq \ell$ and $k_2\coloneq k$. For the other direction,    
    let $\dG$ be a $(k,(s,t),\ell)$-hyperarc-connected orientation of $G$, and let $k_1,k_2\in \Z_{\geq k}$ such that $k_1 + k_2 = \ell+k$. 
    By Menger's theorem, 
    there exist $\ell$ hyperedge-disjoint paths from $s$ to $t$ in $\dG$ and $k$ hyperedge-disjoint paths from $t$ to $s$ in $\dG$. 
    If $k_2=k$, then $\dGp\coloneq\dG$ satisfies the required properties. Suppose $k_2>k$. 
    Let $P_1, \ldots, P_{\ell}$ be $\ell$ hyperedge-disjoint paths from $s$ to $t$ in $\dG$. 
    Consider the orientation $\dGp=(V,E,\head'\colon E\to V)$ obtained from $\dG$ by reversing the orientations of the edges in $P_1, \ldots, P_{k_2-k}$. That is, for each $j\in[q]$, if $P_j$ consists of edges $e_1,\dots,e_q$ where $s\in e_1$, $\head(e_i)\in e_{i+1}$ for $i\in[q-1]$ and $\head(e_q)=t$, then in $\dGp$ we define $\head'(e_1)=s$ and $\head'(e_{i+1})=\head(e_{i})$ for $i\in[q-1]$. For every edge $e$ not contained in the paths, we set $\head'(e)=\head(e)$. 
    
    We claim that $\dGp$ satisfies properties \eqref{prop:kac}-\eqref{prop:k2}. To see this, consider an arbitrary $U\subseteq V$. If $\{s,t\}\not\subseteq U$ or $\{s,t\}\subseteq U$, then every path $P_1, \ldots P_{k_2-k}$ enters and leaves the set $U$ the same number of times and consequently, $d^{in}_{\dGp}(U) = d^{in}_{\dG}(U) \geq k$. If $s\in U$ but $t\notin U$, then $d^{in}_{\dGp}(U) \ge d^{in}_{\dG}(U) + k_2 - k \geq k+(k_2-k)=k_2$. Finally, if $t\in U$ but $s\notin U$, then $d^{in}_{\dGp}(U) \ge d^{in}_{\dG}(U) - (k_2 - k) \geq \ell - (k_2-k) = k+\ell-k_2=k_1$.
    
    We recall that $k_1, k_2\ge k$. Hence, the above observations together imply that $\dGp$ is $k$-hyperarc-connected. Furthermore, by Menger's theorem, $\dGp$ has $k_1$ hyperedge-disjoint paths from $s$ to $t$ and $k_2$ hyperedge-disjoint paths from $t$ to $s$. This concludes the proof of the lemma. 
\end{proof}
    
\end{document}